\documentclass{amsart}              

\usepackage{amsmath}
\usepackage{amsfonts}
\usepackage{amssymb}
\usepackage{mathrsfs}

\usepackage{tikz}
\usetikzlibrary{matrix}

\theoremstyle{plain}
\newtheorem{theorem}{Theorem}[section]

\newtheorem{lemma}[theorem]{Lemma}

\newtheorem{proposition}[theorem]{Proposition}

\newtheorem{definition}[theorem]{Definition}

\newtheorem{assumption}[theorem]{Assumption}
\theoremstyle{remark}
\newtheorem{remark}[theorem]{Remark}
\newtheorem{example}[theorem]{Example}

\numberwithin{equation}{section}

\def\ba{\begin{array}}
\def\ea{\end{array}}
\def\beq{\begin{equation}}
\def\endeq{\end{equation}}
\def\bes{\begin{equation*}}
\def\ees{\end{equation*}}
\def\bea{\begin{eqnarray}}
\def\eea{\end{eqnarray}}
\def\beaa{\begin{eqnarray*}}
\def\eeaa{\end{eqnarray*}}

\def\ol{\overline}
\def\td{\nabla}
\def\pa{\partial}
\def\a{\alpha}
\def\e{\varepsilon}
\def\f{\varphi}
\def\g{\gamma}
\def\si{\sigma}
\def\nts{\negthinspace}
\def\sE{\mathscr{E}}
\def\sL{\mathscr{L}}
\def\cF{\mathcal{F}}
\def\cI{\mathcal{I}}
\def\cT{\mathcal{T}}
\def\hE{\mathbb{E}}
\def\hF{\mathbb{F}}
\def\hP{\mathbb{P}}
\def\hQ{\mathbb{Q}}
\def\hR{\mathbb{R}}
\def\dbL{\rm l\nts L}

\usepackage[colorlinks=true,breaklinks=true,bookmarks=true,urlcolor=blue,
     citecolor=blue,linkcolor=blue,bookmarksopen=false,draft=false]{hyperref}

\begin{document}

	\title[Probability distortion]{Time-consistent Conditional Expectation under Probability Distortion}


	
	\author{Jin Ma}
	\address{Department of Mathematics, University of Southern California, Los Angeles, United States}
	\email{jinma@usc.edu}
	
	\author{Ting-Kam Leonard Wong}
	\address{Department of Statistical Sciences, University of Toronto, Toronto, Ontario, Canada.}
	\email{tkl.wong@utoronto.ca}

	\author{Jiafeng Zhang}
	\address{Department of Mathematics, University of Southern California, Los Angeles, United States}
	\email{jianfenz@usc.edu}
	
	\maketitle
	\date{\today}


\begin{abstract}
We introduce a new notion of conditional nonlinear expectation under probability distortion. Such a distorted nonlinear expectation is not sub-additive in general, so it is beyond the scope of Peng's framework of nonlinear expectations. A more fundamental problem when extending the distorted expectation to a dynamic setting is {\it time-inconsistency}, that is, the usual ``tower property" fails. By localizing the probability distortion and restricting to a smaller class of random variables, we introduce a so-called distorted probability and construct a conditional expectation in such a way that it coincides with the original nonlinear expectation at time zero, but has a time-consistent dynamics in the sense that the tower property remains valid. Furthermore, we show that in the continuous time model this conditional expectation corresponds to a parabolic differential equation whose coefficient involves the law of the underlying diffusion. This work is the first step towards a new understanding of nonlinear expectations under probability distortion, and will potentially be a helpful tool for solving time-inconsistent stochastic optimization problems.
\end{abstract}


\maketitle

%


\section{Introduction.}
In this paper we propose a new notion of nonlinear conditional expectation under {\it probability distortion}. Such a nonlinear expectation is by nature not sub-additive, thus is different from Peng's well-studied nonlinear expectations (see e.g.~\cite{Peng-97, Peng-G}).
Our goal is to find an appropriate definition of  conditional nonlinear  expectations such that it is {\it time-consistent} in the sense that the usual ``tower property" holds.

Probability distortion has been largely motivated by empirical findings in behavioral economics and finance, see, e.g., Kahneman-Tversky \cite{KT1, KT2}, Zhou \cite{Zhou}, and the references therein. It describes the natural human tendency to exaggerate small probabilities for certain events, contradicting the classical axiom of rationality. Mathematically, this can be characterized by a nonlinear expectation where the underlying probability scale is modified by a {\it distortion function}. More precisely, let $\xi$ be a non-negative random variable representing the outcome of an uncertain event. The usual (linear) expectation of $\xi$ can be written in the form
\bea
\label{exp}
\mathbb{E} [\xi] = \int_0^{\infty} \mathbb{P}(\xi \geq x) dx.
\eea
Probability distortion, on the other hand, considers a ``distorted'' version of the expectation
\bea
\label{distoexp}
\sE [\xi] := \int_0^{\infty}  \varphi \big( \mathbb{P}(\xi \geq x)\big) dx,
\eea
where the distortion function $\varphi: [0, 1] \rightarrow [0, 1]$ is  continuous, strictly increasing, and satisfies $\varphi(0) = 0$, $\varphi(1) = 1$. Economically the most interesting case is that $\varphi$ is reverse $S$-shaped, i.e., $\varphi$ is concave when $p\approx 0$ and is convex when $p\approx 1$. In the special case $\varphi(p) \equiv p$, (\ref{distoexp}) reduces to (\ref{exp}). In general the distorted expectation $\sE[\cdot]$ is nonlinear, i.e., neither subadditive nor superadditive.

While \eqref{distoexp} is useful in many contexts, a major difficulty occurs when one tries to define the ``conditional",  or ``dynamic",  
version of the distorted expectation.  Consider, for example,  a ``naively" defined distorted conditional expectation given the information $\mathcal{F}_t$ at time $t$:
\bea
\label{eqn:naive.conditional.expectation}
{\sE_t}[\xi] = \int_0^{\infty} \varphi \left(\mathbb{P}(\xi > x|\cF_t)\right) dx.
\eea
Then it is easy to check that  in general $ \sE_s [ \sE_t [\xi ] ] \neq \sE_s[\xi]$ for $s < t$, i.e., the ``tower property" or the flow property fails. This is often referred to as a type of ``time-inconsistency" and is studied extensively in stochastic optimal control; see Section \ref{sect-inconsistency} for more discussion.

Motivated by the work Karnam-Ma-Zhang \cite{KMZ}, which provides a new perspective for time-inconsistent optimization problems,  in this paper we find a different way to define the distorted conditional expectation so that it remains {\it time-consistent} in terms of preserving the tower property. To be specific, let $(\Omega, \mathcal{F}, \hF, \mathbb{P})$ be a filtered probability space, where $\hF:= \{\mathcal{F}_t\}_{0 \leq t \leq T}$. We look for a family of operators $\{{\sE}_t\}_{0 \leq t \leq T}$ such that for a given $\cF_T$-measurable random variable $\xi$, it holds
that ${\sE}_0 [\xi] = \sE [\xi]$ as in \eqref{distoexp}, and for $0\le s < t\le T$, the tower property holds: $ \sE_s [ \sE_t [\xi ] ] =  \sE_s[\xi]$.  More generally, we shall construct operators $\sE_{s, t}$ for $0 \leq s \leq t \leq T$ such that $\sE_{r, s} [ \sE_{s, t} [ \xi]] = \sE_{r, t}[\xi]$ for $\cF_t$-measurable $\xi$ and $r \leq s \leq t$. We shall argue that this is possible at least for a large class of random variables: $\xi = g(X_t)$, where $g: \mathbb{R} \rightarrow [0, \infty)$ is increasing, and $X$ is either a binomial tree or a one-dimensional diffusion
\bea
\label{X}
dX_t = b(t, X_t) dt + \sigma(t, X_t) dB_t, \quad t\ge 0. 
\eea
It is worth noting that while the aforementioned class of random variables are somewhat restricted, especially the monotonicity of  $g$, which plays a crucial role in our approach (see Remark \ref{rem-nonmonotone}), it contains a large class of practically useful random variables considered in most publications about probability distortion, where $X$ is the state process and $g$ is a utility function, whence  monotone.

The main idea of our approach is based on the following belief: in a dynamic distorted expectation the form of the distortion function should  depend on the prospective time horizon. Simply put, the distortion function over $[0,T]$ such as that in  \eqref{distoexp} is very likely to be different from that in \eqref{eqn:naive.conditional.expectation}, which is applied only to subintervals of $[0, T]$. We believe this is why \eqref{eqn:naive.conditional.expectation} becomes time-inconsistent. Similar to the idea of ``dynamic utility" in \cite{KMZ}, we propose to {\it localize} the distortion function as follows: given a collection of initial distortion functions $\varphi_t$ corresponding to intervals of the form $[0, t]$, we look for  a dynamic distortion function $\Phi(s, t,x; p)$ such that $\Phi(0, t, X_0; \cdot) = \varphi_t$ (e.g. $\varphi_t \equiv \varphi$), and that the resulting distorted conditional expectation 
\begin{equation} 
\label{eqn:naive.conditional.expectation2}
{\sE}_{s, t}[\xi] = \int_0^{\infty} \Phi\left(s, t, X_s; \mathbb{P}(\xi > y|\cF_s)\right) dy, \quad 0\le s<t\le T,
\end{equation}
is time-consistent for all $\xi = g(X_t)$ with  $g$ being increasing. Intuitively, the dependence of the distortion function $\Phi$ on $(s, t,x)$ could be thought of as the agent's (distorted) view  towards the prospective random events at future time $t$ at current time $s$ and state $x$.

We shall first illustrate this idea in discrete time using a binomial tree model to present all main elements. The diffusion case is conceptually similar but the analysis is much more involved. In both cases, however,  the dynamic distortion function has an interesting interpretation: there exists a probability $\mathbb{Q}$ (equivalent to $\mathbb{P}$ and independent of the increasing function $g$) such that
\[
\Phi\big(s,t, x; \mathbb{P}(X_t \geq y|X_s=x)\big) = \mathbb{Q}(X_t \geq y | X_s = x)
\]
(see Theorems \ref{thm-ddf} and \ref{thm-limit} as well as Remark \ref{rem-cEdiscrete}). We shall refer to $\mathbb{Q}$ as the {\it distorted probability}, so that \eqref{eqn:naive.conditional.expectation2} renders the distorted conditional expectation a usual linear conditional expectation under $\mathbb{Q}$. We should note that such a hidden linear structure, due to the restriction $\xi = g(X_t)$, has not been explored in previous works. In particular, in the continuous time setting, this enables us to show that the conditional expectation $\sE_{s,t} [\xi]$ in \eqref{eqn:naive.conditional.expectation2} can be written as $\sE_{s,t} [\xi] = u(s, X_s)$, where the function $u$ satisfies a linear parabolic PDE whose coefficients depend on the distortion function $\varphi$ and the density of the  underlying diffusion  $X$ defined by \eqref{X}. 

We would like to emphasize that while this paper considers only the conditional expectations, it is the first step towards a long term goal of investigating stochastic optimization problems under probability distortion, as well as other time-inconsistent problems.  In fact,  in a very recent paper He-Strub-Zariphopoulou \cite{HSZ} studied an optimal investment problem under probability distortion and showed that a time-consistent dynamic distortion function of the form $\Phi=\Phi(s,t; p)$ exists if and only if it belongs to the family introduced in Wang \cite{Wang} or the agent does not invest in the risky assets. This result in part validates our general framework, which aims at large class of optimization problems of similar type in a  general setting, by allowing $\Phi$ to depend on the state $X_s$, and even its law. 

The rest of the paper is organized as follows. In Section \ref{sect-inconsistency} we review some approaches in the literature for time-inconsistent stochastic optimization problems, which will put this paper in a right perspective.  In Section \ref{sec:distortion} we recall the notion of probability distortion and introduce our dynamic distortion function. In Section \ref{sec:discrete.time} we construct a time-consistent dynamic distortion function in a discrete time binomial tree framework. In Section \ref{sec:continuous.time} we consider the diffusion case \eqref{X} with constant $\sigma$, and the results are extended to the case with general $\sigma$ in Section \ref{sect:generaldiffusion}. Finally, in Section \ref{sect-density} we study the density of the underlying state process $X$, which is crucial for constructing our dynamic distortion function $\Phi$.

\subsection{Discussion: Time-inconsistency in stochastic control.}
\label{sect-inconsistency}

We begin by recalling the usual meaning of  ``time-inconsistency" in a stochastic optimization problem. Consider a stochastic control problem over time horizon $[0, T]$, denote it by $P_{[0,T]}$, and assume $u^*_{0,T}$ is an optimal control. Now for any $t<T$ we consider the same problem over time horizon $[t, T]$ and denote it by 
$P_{[t,T]}$. The dynamic problems $\{P_{[t, T]}\}_{t\in[0,T]}$ is said to be {\it time-consistent} if $u^*_{0,T}\big|_{[t, T]}$ remains optimal for each $P_{[t,T]}$, and {\it time-inconsistent} if it is not. 

Following Strotz \cite{Strotz}, there are two main approaches for dealing with time-inconsistent problems: {\it precommitment strategy} and {\it consistent planning}. The former approach essentially ignores the inconsistency issue and studies only the problem $P_{[0,T]}$, so it can be viewed as a static problem. The consistent planning approach, also known as the {\it game approach}, assumes that the agent plays with future selves and tries to find an  equilibrium. This approach is by nature dynamic, backward in time, and  time-consistent; and the solution is subgame optimal.  Starting from Ekeland-Lazrak \cite{EL},  the game approach has gained strong traction in the math finance community (see e.g.,  Bjork-Murgoci \cite{BM},  Bjork-Murgoci-Zhou \cite{BMZ}, Hu-Jin-Zhou \cite{HJZ},  and Yong \cite{Yong}, to mention a few). We remark, however, that  mathematically the two approaches actually produce different values. 

In Karnam-Ma-Zhang \cite{KMZ} the authors suggested a different perspective. Instead of using a predetermined ``utility" function for all problems $P_{[t,T]}$ as in the game approach (in the context of probability distortion this means using the same $\varphi$ in \eqref{eqn:naive.conditional.expectation} for all $0\le s<t\le T$), in \cite{KMZ} a {\it dynamic utility} is introduced, in the spirit of the {\it predictable forward utility} in
Musiela-Zariphopoulou \cite{MZ1, MZ2} and Angoshtar-Zariphopoulou-Zhou  \cite{AZZ},  to formulate a new dynamic problem $\tilde P_{[t,T]}$, $t\in[0,T]$.  This new dynamic problem is time-consistent and in the meantime $\tilde P_{[0,T]}$ coincides with 
the precommitment $P_{[0,T]}$.  We should note that similar idea also appeared in the works Cui-Li-Wang-Zhu \cite{CLWZ} and Feinstein-Rudloff \cite{FR1, FR2}. In  \cite{KMZ} it is also proposed to use the dynamic programing principle (DPP) to characterize the 
time-consistency, rather than the aforementioned original definition using optimal control. Such a modification is particularly important in situations  where the optimal control does not exist. Noting that the DPP is nothing but the ``tower property" in the absence of   control, we thus consider this paper the first step towards a more general goal.

\section{Static and dynamic probability distortions.}
\label{sec:distortion}

In this section we define probability distortion and introduce the notion of time-consistent dynamic distortion function.

\subsection{Nonlinear expectation under probability distortion.}
Let $(\Omega, \cF, \mathbb{P})$ be a probability space, and let $\dbL^0_+(\cF)$ be the set of $\cF$-measurable random variables $\xi \ge 0$. The notion of probability distortion (see, e.g., Zhou \cite{Zhou}) consists of two elements: (i) a ``distortion function", and (ii) a Choquet-type integral that defines the ``distorted expectation".   More precisely, we have the following definition.

\medskip

\begin{definition} \label{defn-distortion} { \ }
\begin{enumerate}
	\item[(i)] A mapping $\varphi: [0, 1] \rightarrow [0, 1]$ is called a  distortion function if it is continuous, strictly increasing, and satisfies $\varphi(0) = 0$ and $\varphi(1) = 1$.
	\item[(ii)] For any  random variable $\xi \in \dbL^0_+(\cF)$, the distorted expectation operator (with respect to the distortion function $\varphi$) is defined by \eqref{distoexp}. We denote $\dbL^1_\varphi(\cF) := \{\xi\in \dbL^0_+(\cF): \sE[\xi] < \infty\}$.
\end{enumerate}
\end{definition}

\medskip

\begin{remark} { \ }
\begin{enumerate}
	\item[(i)]  The requirement $\xi \ge 0$ is imposed mainly for convenience. 
	\item[(ii)] If  $\varphi(p) = p$, then $\sE[\xi] = \hE^{\hP}[\xi]$ is the standard expectation under $\hP$.
	\item[(iii)] $\sE[\cdot]$ is law invariant, namely $\sE[\xi]$ depends only on the law of $\xi$.
\end{enumerate}
\end{remark}

The following example shows that $\sE$ is in general neither sub-additive nor super-additive. In particular, it is beyond the scope of  Peng \cite{Peng-G} which studies sub-additive nonlinear expectations. 

\medskip

\begin{example} \label{eg-nonlinear}
Assume $\xi_1$ is a Bernoulli random variable: $\hP(\xi_1 = 0) = p,  \hP(\xi_1=1) = 1- p$, and $\xi_2 := 1-\xi_1$. Then clearly $\sE[\xi_1 + \xi_2] = \sE[1]=1$. However, by \eqref{cEdiscrete} below, we have 
\beaa
\sE[\xi_1] =  \varphi(1-p),\quad \sE[\xi_2] = \varphi(p),\quad \mbox{and thus}\quad  \sE[\xi_1] + \sE[\xi_2] = \varphi(p) + \varphi(1-p).
\eeaa
Depending on $\varphi$ and $p$, $\sE[\xi_1] + \sE[\xi_2] $ can be greater than or less than $1$.
\end{example}

\begin{proposition}  \label{prop-DCT}
Assume all the random variables below are in $ \dbL^0_+(\cF)$. Let $c, c_i \ge 0$ be constants. 
\begin{enumerate}
	\item[(i)] $\sE[c] = c$ and $\sE[c\xi] = c \sE[\xi]$.
	\item[(ii)] If $\xi_1 \le \xi_2$, then $\sE[\xi_1] \le \sE[\xi_2]$. In particular, if $c_1 \le \xi \le c_2$,  then $c_1 \le \sE[\xi] \le c_2$. 
	\item[(iii)] Assume $\xi_k$ converges to $\xi$ in distribution, and $\xi^*:=\sup_k \xi_k \in \dbL^1_\varphi(\cF)$. Then $\sE[\xi_k] \to \sE[\xi]$.
\end{enumerate}
\end{proposition}
\begin{proof}  
Since $\varphi$ is increasing,  (i) and (ii) can be verified straightforwardly.  To see (iii), note that $\lim_{k\to \infty}\hP(\xi_k\ge x) = \hP(\xi\ge x)$ for all but  countably many values of $x\in (0,\infty)$. By the continuity of $\varphi$, we have $\lim_{k\to \infty}\varphi(\hP(\xi_k\ge x)) =\varphi( \hP(\xi\ge x))$ for Lebesgue-a.e. $x\in [0,\infty)$. Moreover, since $\varphi$ is increasing, $\varphi(\hP(\xi_k\ge x)) \le \varphi(\hP(\xi^*\ge x))$ for all $k$. By \eqref{distoexp} and the dominated convergence theorem we have $\sE[\xi_k] \to \sE[\xi]$. 
\end{proof} 

We now present two special cases that will play a crucial role in our analysis. In particular, they will lead naturally to the concept of {\it distorted probability}.  Let
\bea
\label{dbI}
\cI := \{g: \hR \to [0,\infty):  \mbox{$g$ is bounded, continuous, and increasing}\}.
\eea

\begin{proposition} \label{prop-discrete} { \ }
\begin{enumerate}
	\item[(i)] Assume $\eta \in \dbL^1_\varphi(\cF)$ takes only finitely many values $x_1,\cdots, x_n$. Then
	\bea
	\label{cEdiscrete0}
	\sE[ \eta ] = \sum_{k=1}^n x_{(k)} \left[\varphi\left(\hP(\eta \ge x_{(k)})\right) - \varphi\left(\hP(\eta \geq x_{(k+1)})\right)\right],
	\eea
	where $x_{(1)}\le \cdots\le x_{(n)}$ are the ordered values of $x_1,\cdots, x_n$, and   $x_{(n+1)} := \infty$.
	
	In particular, if $x_1<\cdots< x_n$ and $g\in \cI$, then
	\bea
	\label{cEdiscrete}
	\sE[g(\eta)] = \sum_{k=1}^n g(x_k) \left[\varphi\left(\hP(\eta \ge x_{k})\right) - \varphi\left(\hP(\eta\ge x_{k+1})\right)\right].
	\eea

	\item[(ii)] Assume $\eta \in \dbL^0(\cF)$ has density $\rho$, and $g\in \cI$, $\varphi \in C^1([0, 1])$.  Then
	\bea
	\label{cEcont}
	\sE[g(\eta)] = \int_{-\infty}^\infty g(x) \rho(x) \varphi'(\hP(\eta\ge x))dx.
	\eea
\end{enumerate}
\end{proposition}

\begin{proof}  
 (i) Denote $x_{(0)} := 0$. It is clear that $\hP(\eta \ge x) = \hP(\eta\ge x_{(k)})$  for $ x\in (x_{(k-1)}, x_{(k)}]$. Then
\beaa
\sE[\eta] =\int_0^{\infty} \varphi(\hP(\eta \geq x)) dx=\sum_{k=1}^n [x_{(k)} - x_{(k-1)}] \varphi(\hP(\eta \geq x_{(k)}),
\eeaa
which implies \eqref{cEdiscrete} by using a simple Abel rearrangement as well as the fact $\varphi(\hP(\eta \ge x_{(n+1)})) =0$. 

(ii)  We proceed in four steps.

{\it Step 1.} Assume $g$ is bounded, strictly increasing,  and differentiable. Let $a :=  g(-\infty), b :=  g(\infty)$. Then, 
$\varphi(\hP(g(\eta)\ge x))=1$, $x\le a$; $\varphi(\hP(g(\eta)\ge x))=0$, $x\ge b$, and integration by parts yields
\begin{eqnarray}
\label{sEgsmooth}
\sE[g(\eta)] &=& a + \int_a^b \varphi(\hP(g(\eta) \ge x)) dx = a + \int_{-\infty}^\infty \varphi(\hP(\eta \ge x)) g'(x) dx\\
&=& a + \left.\varphi(\hP(\eta \ge x)) g(x)\right|_{x=-\infty}^{x=\infty} - \int_{-\infty}^\infty g(x)  {d\over dx}\left(\varphi(\hP(\eta \ge x))\right) dx\nonumber\\
&=& \int_{-\infty}^\infty g(x)  \rho(x) \varphi'(\hP(\eta \ge x)) dx.\nonumber
\end{eqnarray}

{\it Step 2.} Assume $g$ is bounded, increasing,  and continuous.  One can easily construct $g_n$ such that each $g_n$ satisfies the requirements in Step 1 and $g_n$ converges to $g$ uniformly. By Step 1, \eqref{cEcont} holds for each $g_n$. Send $n\to \infty$ and apply Proposition \ref{prop-DCT} (iii) we prove \eqref{cEcont} for $g$. 

{\it Step 3.} Assume $g$ is increasing and bounded by a constant $C$. For any $\e>0$, one can construct a continuous and increasing function $g_\e$ and an open set $O_\e$ such that $|g_\e|\le C$, $|g_\e(x) - g(x)|\le \e$ for $x\notin O_\e$,   and the Lebesgue measure $|O_\e|\le \e$.  Then \eqref{cEcont} holds for each $g_\e$. Note that
\beaa
\hE\left[|g_\e(\eta) - g(\eta)|\right] \le \e + 2C \hP(\eta\in O_\e) = \e + 2C \int_{O_\e} \rho(x) dx \to 0\quad \mbox{as}~ \e\to 0.
\eeaa
Then $g_\e(\eta) \to g(\eta)$ in distribution and thus $\sE[g_\e(\eta)]\to \sE[g(\eta)]$ by Proposition \ref{prop-DCT} (iii). Similarly,
\beaa
\int_{-\infty}^\infty |g_\e(x)-g(x)| \rho(x) \varphi'(\hP(\eta\ge x))dx \le \e + 2C \int_{O_\e} \rho(x) \varphi'(\hP(\eta\ge x))dx  \to 0.
\eeaa
Then we obtain \eqref{cEcont} for $g$. 

{\it Step 4.} In the general case, denote $g_n := g \wedge n$. Then \eqref{cEcont} holds for each $g_n$ and $g_n \uparrow g$. By monotone convergence theorem, 
\beaa
\lim_{n\to\infty}  \int_{-\infty}^\infty g_n(x) \rho(x) \varphi'(\hP(\eta\ge x))dx =  \int_{-\infty}^\infty g(x) \rho(x) \varphi'(\hP(\eta\ge x))dx.
\eeaa
If $g(\eta) \in \dbL^1_\varphi(\cF)$, then by Proposition \ref{prop-DCT} (iii) we obtain \eqref{cEcont} for $g$. Now assume $\sE[g(\eta)]=\infty$. Following the arguments in  Proposition \ref{prop-DCT} (iii), note that $\hP(g_n(\eta) \ge x) \uparrow \hP(g(\eta)\ge x)$ for Lebesgue-a.e. $x\in[0,\infty)$, as $n\to\infty$. Then by monotone convergence theorem one can verify that $\sE[g_n(\eta)]=\int_0^\infty \varphi(\hP(g_n(\eta) \ge x)) dx ~ \uparrow ~ \int_0^\infty \varphi(\hP(g(\eta) \ge x)) dx=\sE[g(\eta)]$, proving \eqref{cEcont} again.  
\end{proof}

\medskip

\begin{remark} \label{rem-cEdiscrete} { \ }
\begin{enumerate}
	\item[(i)] In the discrete case, the formula \eqref{cEdiscrete} can be interpreted as follows. For each $k$, define the {\it distorted probability} $q_{k}$ by
	\bea
	\label{distprob}
	q_k := \varphi(\hP(\eta \geq x_{k})) - \varphi(\hP(\eta \geq x_{k + 1})),\quad k = 1, 2, \ldots, n.
	\eea
	Then $q_k \ge 0$, $\sum_{k=1}^n q_k = 1$, and  ${\sE}[g(\eta)] = \sum_{k = 1}^n g(x_{k}) q_{k}$. So $\{q_k\}$ plays the role of a ``probability distribution", and $\sE$ is the usual linear expectation under the (distorted) probability $\{q_k\}$. This observation will be the foundation of  our analysis below.
	\item[(ii)] In the continuous case, the situation is similar. Indeed, denote $\widetilde \rho(x) := \rho(x) \varphi'(\hP(\eta \ge x))$. Then $\widetilde \rho$ is also a density function, and by \eqref{cEcont}, $\sE[g(\eta)] = \int_{-\infty}^\infty g(x) \widetilde \rho(x) dx$ is the usual expectation under the distorted density $\widetilde \rho$ of $\eta$. 
	\item[(iii)] Although the operator $\sE: \dbL^1_\varphi(\cF)\to [0, \infty)$ is nonlinear in general,  for fixed $\eta$, the restricted mapping $g\in\cI \mapsto \sE[g(\eta)]$ is linear  under non-negative linear combinations. 
	\item[(iv)] Actually, for any $\xi\in \dbL_\varphi^1(\cF)$, note that $F_\xi(x) := 1- \varphi(\hP(\xi\ge x))$, $x\ge 0$, is a cdf, and thus defines a distorted probability measure $\hQ^\xi$ such that $\sE[\xi] = \hE^{\hQ^\xi}[\xi]$. However, this $\hQ^\xi$ depends on $\xi$. The main feature in \eqref{cEdiscrete} and \eqref{cEcont} is that, for a given $\eta$, we find a common distorted probability measure for all $\xi\in \{g(\eta): g\in \cI\}$.
\end{enumerate}
\end{remark}

\subsection{Time-inconsistency.}
Let $0\in \cT \subset [0, \infty)$ be the set of possible times, and $X=\{X_t\}_{t\in \cT}$ be a Markov process with deterministic $X_0$.
Denoting $\hF = \{\cF_t\}_{t\in \cT}= \hF^X$ be the filtration generated by $X$,  we want to define an $\cF_t$-measurable conditional expectation $\sE_t[\xi]$ such that each $\sE_t[\xi]$ is $\cF_t$-measurable, and the following ``tower property" (or ``flow property") holds (we will consider $\sE_{s, t}$ later on):
\bea
\label{flow}
\sE_s \left[\sE_t[\xi]\right] = \sE_s[\xi],\quad \mbox{for all}~ s, t \in \cT \text{ such that }  s < t.
\eea

We note that the tower property (\ref{flow}) is standard for the usual (linear) expectation as well as the sub-linear $G$-expectation of Peng \cite{Peng-G}. It is also a basic requirement of the so-called {\it dynamic risk measures} (see e.g. Bielecki-Cialenco-Pitera \cite{BCP}). However, under probability distortion, the simple-minded definition of the conditional expectation given by \eqref{eqn:naive.conditional.expectation} could very well be time-inconsistent. Here is a simple explicit example:

\medskip

\begin{example} \label{eg-inconsistent}
Consider a two period binomial tree model: $X_t = \sum_{i=1}^t \zeta_i$, $ t\in \cT :=\{0,1,2\}$, where $\zeta_1$, $\zeta_2$ are independent Rademacher random variables with  $\hP(\zeta_i=\pm 1)  = {1\over 2}$, $i=1,2$.  Let $\varphi(p) := p^2$, $\xi := g(X_2)$ for some strictly increasing function $g$, and $\sE_1[\xi]$ be defined by \eqref{eqn:naive.conditional.expectation}. Then 
\bea
\label{inconsistent}
\sE\left[\sE_1[\xi]\right] \neq \sE[\xi].
\eea 
\end{example}
\begin{proof}  
By \eqref{cEdiscrete}, we have
\beaa
\left.\sE_1[\xi]\right|_{X_1=-1} = g(-2) \left[1-\varphi({1\over 2})\right] + g(0) \varphi({1\over 2}),\quad  \left.\sE_1[\xi]\right|_{X_1=1} = g(0) \left[1-\varphi ({1\over 2})\right] + g(2) \varphi ({1\over 2}).
\eeaa
Note that $\sE_1[\xi]\big|_{X_1=-1}  <  \sE_1[\xi]\big|_{X_1=1}$ since $g$ is strictly increasing. Then,  by \eqref{cEdiscrete} again, we have
\begin{eqnarray}
\label{inconsistentEExi}
&&\sE\left[ \sE_1[\xi]\right] =  \left.\sE_1[\xi]\right|_{X_1=-1}  \left[1-\varphi ({1\over 2})\right]  +  \left.\sE_1[\xi]\right|_{X_1=1} \varphi ({1\over 2})\\
&=& g(-2) \left[1-\varphi ({1\over 2}) \right]^2 + 2 g(0) \varphi ({1\over 2} )\left[1-\varphi ({1\over 2})\right] + g(2) \left[\varphi({1\over 2})\right]^2=
{9\over 16} g(-2) + {3\over 8} g(0) + {1\over 16} g(2). \nonumber
\end{eqnarray}
On the other hand, by \eqref{cEdiscrete} we also have
\bea
\label{inconsistentExi}
\sE[\xi] =  g(-2) \left[1- \varphi ({3\over 4}) \right] + g(0) \left[\varphi ({3\over 4}) - \varphi ({1\over 4})\right] + g(2) \varphi ({1\over 4})= {7\over 16} g(-2) + {1\over 2} g(0) + {1\over 16} g(2).
\eea
Comparing (\ref{inconsistentEExi}) and (\ref{inconsistentExi}) and noting that $g(-2) < g(0)$, we obtain $\sE\left[ \sE_1[\xi]\right]  <  \sE[\xi]$.  
\end{proof}

\subsection{Time-consistent dynamic distortion function.}
As mentioned in the Introduction, an apparent reason for the time-inconsistency of the ``naive" distorted conditional  expectation \eqref{eqn:naive.conditional.expectation} is that the distortion function $\varphi$ is time-invariant.  Motivated by the idea of {\it dynamic utility} in Karnam-Ma-Zhang \cite{KMZ}, we introduce the notion of {\it time-consistent dynamic distortion function} which forms the framework of this paper.  Denote
$$
\cT_2 := \{(s, t) \in \cT \times \cT: s< t\}.
$$

\begin{definition}    \label{defn-ddf} {\ }
\begin{enumerate}
	\item[(i)] A  mapping  $\Phi: \cT_2 \times \hR \times [0, 1] \to [0, 1]$ is called a dynamic distortion function if it is jointly Lebesgue measurable in $(x, p)$ for any $(s, t)\in\cT_2$ and, for each $(s, t,x)\in \cT_2\times \hR$, the mapping $p\in [0,1] \mapsto \Phi(s, t,x; p)$ is a distortion function in the sense of Definition \ref{defn-distortion}.
	\item[(ii)] Given a dynamic distortion function $\Phi$,  for any $(s, t)\in \cT_2$ we define $\sE_{s,t}$ as follows:
	\bea
	\label{condcE}
	\sE_{s,t}[\xi] := \int_0^\infty \Phi(s, t,X_s; \hP(\xi\ge x|\cF_s)) dx,\quad \xi\in \dbL^0_+(\sigma(X_t)).
	\eea
	\item[(iii)] We say a dynamic distortion function $\Phi$ is time-consistent if the tower property holds:
	\begin{equation}
	\label{flow2}
	\sE_{r,t}[g(X_t)] =  \sE_{r,s}\left[ \sE_{s,t}[g(X_t)]\right],\quad r,s,t\in \cT, ~0\le r<s<t\le T, ~g\in \cI. 
	\end{equation}
\end{enumerate}
\end{definition}

\medskip

\begin{remark} \label{rem-ddf} { \ }
	\begin{enumerate}
		\item[(i)] Compared to the naive definition \eqref{eqn:naive.conditional.expectation}, the dynamic distortion function in \eqref{condcE} depends also on the current time 
		$s$, the ``terminal" time $t$, and the current state $x$.  This enables us to describe  different (distorted) perceptions of   future events at different times and states.  For example, people may feel very differently towards a catastrophic event that might happen tomorrow as opposed to ten years later with the same probability.  
		\item[(ii)] In this paper we apply $\sE_{s, t}$ only on $\xi=g(X_t)$ for some $g\in \cI$. As we saw in Remark \ref{rem-cEdiscrete}, in this case the operator $\sE_{s,t}$ will be linear in $g$. The general case with non-monotone $g$ (or even path dependent $\xi$)  seems to be very challenging and will be left to  future research, see Remark \ref{rem-nonmonotone} below. It is worth noting, however, that in many applications $g$ is a utility function, which is indeed increasing. 
		\item[(iii)] Given $g\in \cI$, one can easily show that $\sE_{s,t}[g(X_t)] = u(s, X_s)$ for some function $u(s,\cdot)\in \cI$. This justifies the right side of \eqref{flow2}.
	\end{enumerate}
\end{remark}

\medskip

Now, for each $0<t\in \cT$, we assume that an initial distortion function $\varphi_t(\cdot)$ is given (a possible choice is $\varphi_t \equiv \varphi$) as the perspective at time $0$ towards the future events at $t>0$.   Our goal is  to construct a time-consistent dynamic distortion function $\Phi$ such that $\Phi(0, t, X_0; \cdot) = \varphi_t(\cdot)$ for all $0<t\in \cT$. We shall consider models both in discrete time and in continuous time.

\section{The binomial tree case.} 
\label{sec:discrete.time}
In this section we consider a binomial tree model which contains all the main ideas of our approach. Let $\{\varphi_t\}_{t\in \cT\backslash \{0\}}$  be a given family of initial distortion functions.

\subsection{The two-period binomial tree case.} 
\label{sec:2period}

To illustrate our main idea, let us first consider the simplest case when $X$ follows a two-period binomial tree as in Example \ref{eg-inconsistent} (see the left graph in Figure \ref{fig:binomial_tree1}). Let  $\xi = g(X_2)$ where $g\in \cI$.  We shall construct $\Phi(1, 2, x; p)$ and $\sE_{1,2}[\xi]$.  

Note that $\Phi(0,t, 0; \cdot) = \varphi_t(\cdot)$ for $t=1,2$, by  \eqref{cEdiscrete} we have
\begin{equation}
\label{cEX2}
{\sE}_{0,2} [\xi] = g(-2) \left[\varphi_2(1) - \varphi_2(\frac{3}{4})\right] + g(0) \left[\varphi_2 (\frac{3}{4} ) - \varphi_2 (\frac{1}{4} ) \right] + g(2)  \left[\varphi_2 (\frac{1}{4} ) - \varphi_2(0) \right].
\end{equation}
Here we write $\varphi_2(0)$ and $\varphi_2(1)$, although their values are $0$ and $1$, so that formula \eqref{Phi1} below will be more informative when extending to multi-period models. Assume $\sE_{1,2}[\xi] = u(1, X_1)$. Then by definition we should have
\begin{eqnarray}
\label{u1}
u(1, -1) &=& g(-2) \left[1 - \Phi(1, 2, -1; {1\over 2})\right] + g(0) \Phi(1,2, -1; {1\over 2}),\\
u(1, 1) &=& g(0) \left[1 - \Phi(1, 2,1; {1\over 2})\right] + g(2) \Phi(1, 2, 1; {1\over 2}).
\end{eqnarray}
Assume now that $u(1,\cdot)$ is also increasing,  then by  \eqref{cEdiscrete} again we have
\begin{equation}
\label{cEX1}
{\sE}_{0,1} \left[\sE_{1,2}[\xi]\right] = \sE_{0,1}\left[u(1, X_1)\right]  = u(1, -1) \left[\varphi_1(1) - \varphi_1(\frac{1}{2})\right] + u(1,1) \left[\varphi_1 (\frac{1}{2} ) - \varphi_1(0) \right].
\end{equation}
Plugging \eqref{u1} into \eqref{cEX1}:
\beaa
{\sE}_{0,1} \left[\sE_{1,2}[\xi]\right] &=& g(-2) \left[1- \Phi(1, 2, -1; {1\over 2})\right] \left[\varphi_1(1)- \varphi_1(\frac{1}{2})\right]  +  g(2)  \Phi(1, 2, 1; {1\over 2})[\varphi_1 (\frac{1}{2} )-\varphi_1(0)]\\
&&+ g(0) \left[ \Phi(1,2, -1; {1\over 2}) \left[\varphi_1(1)- \varphi_1(\frac{1}{2})\right]   + \left[1- \Phi(1, 2, 1; {1\over 2})\right]  [\varphi_1 (\frac{1}{2} )-\varphi_1(0)]\right].
\eeaa
Recall from \eqref{flow2} that we want  the above to be equal to \eqref{cEX2} for all $g\in \cI$. This leads to a natural and unique choice:
\bea
\label{Phi1}
\Phi(1, 2,-1; {1\over 2}) := \frac{\varphi_2(\frac{3}{4}) - \varphi_1(\frac{1}{2})}{\varphi_1(1) - \varphi_1(\frac{1}{2})},\quad  \Phi(1, 2, 1; {1\over 2}):=  \frac{\varphi_2(\frac{1}{4}) - \varphi_1(0)}{\varphi_1(\frac{1}{2}) - \varphi_1(0)}.
\eea
Consequently, \eqref{u1} now reads 
\bea
\label{u2}
u(1, -1) &=& g(-2) \left[1-\Phi(1, 2,-1; {1\over 2})\right] + g(0)\Phi(1, 2,-1; {1\over 2}),\\
u(1, 1) &=& g(0) \left[1-\Phi(1, 2, 1; {1\over 2})\right] + g(2) \Phi(1, 2, 1; {1\over 2}).
\eea
Note that since $\varphi_2(\cdot)$ is strictly increasing. Assuming further $\varphi_2({1\over 4}) < \varphi_1({1\over 2}) < \varphi_2({3\over 4})$ and using \eqref{Phi1}, we have  
\bea
\label{mon1}
0< \Phi(1, 2,-1; {1\over 2})< 1,\quad 0< \Phi(1, 2, 1; {1\over 2})< 1.
\eea
Note that  (\ref{u2}) and \eqref{mon1} imply that $u(1,-1) \le g(0) \le u(1, 1)$, thus $u(1,\cdot)$ is indeed increasing. 

Finally, we note that the distorted  expectations $\sE_{0,1}[u(1, X_1)]$,  $\sE_{0,2}[g(X_2)]$, and the distorted conditional expectation $\sE_{1,2}[g(X_2)]$ can be viewed as  a standard  expectation and conditional expectation, but under a new {\it distorted probability measure} described in the right graph in Figure \ref{fig:binomial_tree1}, where 
\begin{equation}
\label{qij2}
q_{0,0}^+ := \varphi_1({1\over 2}),\quad q_{1,1}^+ := \frac{\varphi_2(\frac{1}{4}) - \varphi_1(0)}{\varphi_1(\frac{1}{2}) - \varphi_1(0)},\quad q_{1,0}^+ := \frac{\varphi_2(\frac{3}{4}) - \varphi_1(\frac{1}{2})}{\varphi_1(1) - \varphi_1(\frac{1}{2})},\quad q_{i,j}^- := 1- q_{i,j}^+.
\end{equation}
This procedure resembles finding the risk-neutral measure in  option pricing theory, whereas the arguments of $\varphi_t$ in \eqref{qij2} represent the quantiles of the simple random walk. 

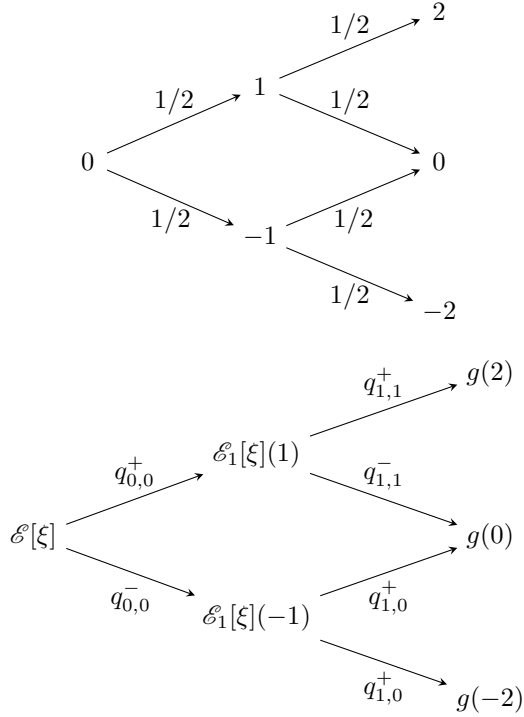
\begin{figure}[t!]
	\centering
	\begin{tikzpicture}[>=stealth]
	\matrix (tree) [
	matrix of nodes,
	minimum size=0.5cm,
	column sep=1.7cm,
	row sep=0.5cm,
	]
	{
		&              & $2$ \\
		&  $1$ &   \\
		$0$ &               & $0$ \\
		& $-1$ &   \\
		&              & $-2$ \\
	};
	\draw[->] (tree-3-1) -- (tree-2-2) node [midway,above] {$1/2$};
	\draw[->] (tree-3-1) -- (tree-4-2) node [midway,below] {$1/2$};
	\draw[->] (tree-2-2) -- (tree-1-3) node [midway,above] {$1/2$};
	\draw[->] (tree-2-2) -- (tree-3-3) node [midway,above] {$1/2$};
	\draw[->] (tree-4-2) -- (tree-3-3) node [midway,below] {$1/2$};
	\draw[->] (tree-4-2) -- (tree-5-3) node [midway,below] {$1/2$};
	\end{tikzpicture}\quad
	\begin{tikzpicture}[>=stealth]
	\matrix (tree) [
	matrix of nodes,
	minimum size=0.5cm,
	column sep=1.7cm,
	row sep=0.5cm,
	]
	{
		&              & $g(2)$ \\
		&  ${\sE}_1 [\xi](1)$ &   \\
		$ \sE[ \xi]$ &               & $g(0)$ \\
		& ${\sE}_1 [\xi](-1)$ &   \\
		&              & $g(-2)$ \\
	};
	\draw[->] (tree-3-1) -- (tree-2-2) node [midway,above] {$q_{0,0}^+$};
	\draw[->] (tree-3-1) -- (tree-4-2) node [midway,below] {$q_{0,0}^-$};
	\draw[->] (tree-2-2) -- (tree-1-3) node [midway,above] {$q_{1,1}^+$};
	\draw[->] (tree-2-2) -- (tree-3-3) node [midway,above] {$q_{1,1}^-$};
	\draw[->] (tree-4-2) -- (tree-3-3) node [midway,below] {$q_{1,0}^+$};
	\draw[->] (tree-4-2) -- (tree-5-3) node [midway,below] {$q_{1,0}^+$};
	\end{tikzpicture}
	\caption{Two period binomial tree: left for $X$ and right for ${\sE}_t [\xi]$, with $(q_{i,j}^+, q_{i,j}^-)$ in \eqref{qij2}.} 
	\label{fig:binomial_tree1}
\end{figure}

\medskip

\begin{remark} \label{rem-nonmonotone} { \ }
	\begin{enumerate}
		\item[(i)] We now explain why it is crucial to assume  $g\in \cI$. Indeed, assume instead that $g$ is decreasing. Then by \eqref{cEdiscrete0} and following similar arguments we can see that
		\beaa
		\Phi(1, 2,1; {1\over 2}) = \frac{\varphi_2(\frac{3}{4}) - \varphi_1(\frac{1}{2})}{\varphi_1(1) - \varphi_1(\frac{1}{2})},\quad  \Phi(1, 2, -1; {1\over 2})=  \frac{\varphi_2(\frac{1}{4}) - \varphi_1(0)}{\varphi_1(\frac{1}{2}) - \varphi_1(0)}.
		\eeaa
		This is in general different from \eqref{Phi1}. That is, we cannot find a common time-consistent dynamic distortion function which works for both increasing and  decreasing functions $g$.
		
		\item[(ii)] For a fixed (possibly non-monotone) function $g: \hR \to [0, \infty)$, it is possible to construct $\Phi$ such that $\sE_{0,2}[g(X_2)] = \sE_{0,1}[\sE_{1,2}[g(X_2)]]$. However, this $\Phi$ may depend on $g$. It seems to us that this is too specific and thus is not desirable.  
		
		\item[(iii)] Another challenging case is when $X$ has crossing edges. This destroys the crucial monotonicity in a different way and $\Phi$ may not exist, as we shall see in Example \ref{eg-crossing} below. There are two ways to understand the main difficulty here:  for the binary tree in Figure \ref{fig:binary_tree} and for $g\in \cI$, 
		\begin{itemize}
			\item $u(1, -1)$ is the weighted average of $g(-2)$ and $g(1)$, and $u(1, 1)$ is the weighted average of $g(-1)$ and $g(2)$. Since $g(-1)<g(1)$, for any given $\Phi$, there exists some $g\in \cI$ such that $u(1, -1) > u(1,1)$, namely $u(1,\cdot)$ is not increasing in $x$.
			\item In $\sE_{1,2}[g(X_2)]$ the conditional probability $p_2 = \hP(X_2=1|X_1=-1)$ would contribute to the weight of $g(1)$, but not to that of $g(-1)$. However, since $g(-1) <g(1)$, in $\sE_{0,2}[g(X_2)]$ the $p_2$ will contribute to the weight of $g(1)$ as well. This discrepancy destroys the tower property.  The same issue also arises in continuous time when the diffusion coefficient is non-constant; see Remark \ref{rem-general} below.
		\end{itemize}
	\end{enumerate}
\end{remark}

\medskip

The following example shows that it is essential to require that the tree is recombining.

\medskip

\begin{example}
	\label{eg-crossing}
	Assume $X$ follows the binary tree in Figure \ref{fig:binary_tree} and  $g\in \cI$ is strictly increasing. Then in general there is no time-consistent dynamic distortion function $\Phi$.
	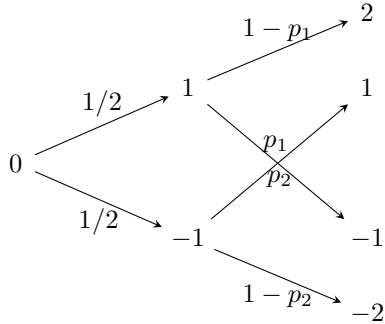
\begin{figure}[h]
		\centering
		\begin{tikzpicture}[>=stealth]
		\matrix (tree) [
		matrix of nodes,
		minimum size=0.5cm,
		column sep=1.7cm,
		row sep=0.5cm,
		]
		{
			&              & $2$ \\
			&  $1$ &   $1$\\
			$0$ &               &  \\
			& $-1$ &   $-1$\\
			&              & $-2$ \\
		};
		\draw[->] (tree-3-1) -- (tree-2-2) node [midway,above] {$1/2$};
		\draw[->] (tree-3-1) -- (tree-4-2) node [midway,below] {$1/2$};
		\draw[->] (tree-2-2) -- (tree-1-3) node [midway,above] {$1-p_1$};
		\draw[->] (tree-2-2) -- (tree-4-3) node [midway,above] {$p_1$};
		\draw[->] (tree-4-2) -- (tree-2-3) node [midway,below] {$p_2$};
		\draw[->] (tree-4-2) -- (tree-5-3) node [midway,below] {$1-p_2$};
		\end{tikzpicture}
		\caption{A two period binary tree with crossing edges.} 
		\label{fig:binary_tree}
	\end{figure}
\end{example}
\begin{proof}  
By \eqref{cEdiscrete0} we have
\beaa
\sE_{0,2}[g(X_2)] &=&  g(-2)[1-\varphi_2({1+p_2\over 2})] + g(-1) [\varphi_2({1+p_2\over 2}) - \varphi_2({1-p_1+p_2\over 2})] \\
&&+ g(1) [\varphi_2({1-p_1+p_2\over 2}) - \varphi_2({1-p_1\over 2})] + g(2) \varphi_2({1-p_1\over 2}).
\eeaa 
Assume $\sE_{1,2}[g(X_2)] = u(1, X_1)$. Then by definition we should have
\beaa
u(1, -1) &=& g(-2) \left[1 - \Phi(1, 2, -1; p_2)\right] + g(1)\Phi(1,2, -1; p_2),\\
u(1, 1) &=& g(-1) \left[1 - \Phi(1, 2,1; 1-p_1)\right] + g(2)\Phi(1, 2, 1; 1-p_1).
\eeaa
Assume without loss of generality that $u(1,-1)<u(1,1)$, and the case $u(1,1)<u(1,-1)$ can be analyzed similarly. Then 
\beaa
{\sE}_{0,1} \left[\sE_{1,2}[g(X_2)]\right] &=& g(-2) \left[1- \Phi(1, 2, -1; p_2)\right] \left[1- \varphi_1(\frac{1}{2})\right]  +   g(1) \Phi(1, 2, -1; p_2) \left[1- \varphi_1(\frac{1}{2})\right]\\
&&+ g(-1) \left[1 - \Phi(1, 2,1; 1-p_1)\right] \varphi_1 (\frac{1}{2} )+ g(2)\Phi(1, 2, 1; 1-p_1)  \varphi_1 (\frac{1}{2} ).
\eeaa
If the tower property holds: $\sE_{0,2}[g(X_2)] = {\sE}_{0,1} \big[\sE_{1,2}[g(X_2)]\big]$ for all $g\in \cI$, comparing the weights of $g(-1)$ and $g(2)$ we have
\beaa
&\big[1 - \Phi(1, 2,1; 1-p_1)\big] \varphi_1 (\frac{1}{2} ) =  \varphi_2({1+p_2\over 2}) - \varphi_2({1-p_1+p_2\over 2}),\\
& \Phi(1, 2, 1; 1-p_1)  \varphi_1 (\frac{1}{2} ) = \varphi_2({1-p_1\over 2}).
\eeaa
Adding the two terms above, we have
\bea
\label{phi12constraint}
\varphi_1(\frac{1}{2}) = \varphi_2({1+p_2\over 2}) - \varphi_2({1-p_1+p_2\over 2}) + \varphi_2({1-p_1\over 2}).
\eea
This equality does not always hold. In other words, unless $\varphi$ satisfies \eqref{phi12constraint}, there is no time-consistent $\Phi$ for the model in Figure \ref{fig:binary_tree}.  
\end{proof}

\subsection{The general binomial tree case.}
\label{sect:GBT}

We now extend our idea to a general binomial tree model. Let $\cT$ consists of the points $0=t_0<\cdots<t_N$, and let $X=\{X_{t_i}\}_{0\le i\le N}$ be a finite state Markov process such that for each $i=0,\ldots, N$, $X_{t_i}$ takes values $x_{i,0}<\cdots<x_{i,i}$, and has the following transition probabilities: 
\begin{equation}
\label{pijgeneral}
\hP\big(X_{t_{i+1}} = x_{i+1, j+1} \big| X_{t_i} = x_{i,j}\big) = p_{i,j}^+,\quad \hP\big(X_{t_{i+1}} = x_{i+1, j} \big| X_{t_i} = x_{i,j}\big) = p_{i,j}^- := 1- p_{i,j}^+,
\end{equation}
where $p_{ij}^{\pm} > 0$. See Figure \ref{fig:general.tree1} for the case $N=3$. We also assume that for each $t_i \in \mathcal{T} \setminus \{0\}$ we are given a distortion function $\varphi_{t_i}$.

\begin{figure}[t!]
	\centering
	\begin{tikzpicture}[>=stealth]
	\matrix (tree) [
	matrix of nodes,
	minimum size=0.5cm,
	column sep=1.7cm,
	row sep=0.5cm,
	]
	{
		&                 &                  &$x_{3,3}$ \\
		&                 &$x_{2,2}$&                  \\
		&$x_{1,1}$&                  &$x_{3,2}$ \\
		$x_{0,0}$&                 &$x_{2,1}$&                   \\
		&$x_{1,0}$&                  &$x_{3,1}$ \\
		&                 &$x_{2,0}$&                   \\
		&                 &                  &$x_{3,0}$ \\
	};
	\draw[->] (tree-4-1) -- (tree-3-2) node [midway,above] {$p_{0,0}^+$};
	\draw[->] (tree-4-1) -- (tree-5-2) node [midway,below] {$p_{0,0}^-$};
	\draw[->] (tree-3-2) -- (tree-2-3) node [midway,above] {$p_{1,1}^+$};
	\draw[->] (tree-3-2) -- (tree-4-3) node [midway,above] {$p_{1,1}^-$};
	\draw[->] (tree-5-2) -- (tree-4-3) node [midway,below] {$p_{1,0}^+$};
	\draw[->] (tree-5-2) -- (tree-6-3) node [midway,below] {$p_{1,0}^-$};
	\draw[->] (tree-2-3) -- (tree-1-4) node [midway,above] {$p_{2,2}^+$};
	\draw[->] (tree-2-3) -- (tree-3-4) node [midway,above] {$p_{2,2}^-$};
	\draw[->] (tree-4-3) -- (tree-3-4) node [midway,above] {$p_{2,1}^+$};
	\draw[->] (tree-4-3) -- (tree-5-4) node [midway,below] {$p_{2,1}^-$};
	\draw[->] (tree-6-3) -- (tree-5-4) node [midway,below] {$p_{2,0}^+$};
	\draw[->] (tree-6-3) -- (tree-7-4) node [midway,below] {$p_{2,0}^-$};
	\end{tikzpicture}
	
	\caption{Three period binomial tree for $X$} 
	\label{fig:general.tree1}
\end{figure}
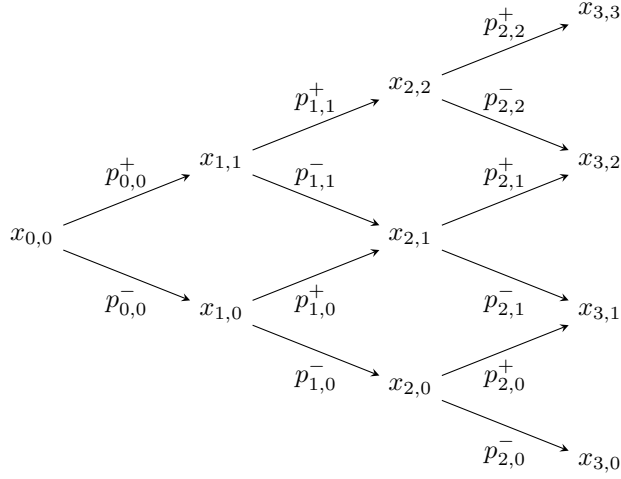

Motivated by the analysis in Section \ref{sec:2period}, we shall find a distorted probability measure $\hQ$ so that 
\bea
\label{sEQ}
\sE_{s, t}[g(X_t)] = \hE^{\hQ}[g(X_t)|\sigma(X_s)]\quad \mbox{ for all $g\in \cI$}.
\eea
This  implies the tower property of $\sE_{s,t}$ immediately and naturally leads to a time-consistent dynamic distortion function. Keeping \eqref{qij2} in mind, we define the following distorted probabilities for the binomial tree model: for $0\le j\le i\le N$, 
\begin{equation}
\label{qij}
q_{i, j}^+ :=  \frac{\varphi_{t_{i+1}}(G_{i+1, j+1} ) - \varphi_{t_i}(G_{i, j + 1})}{\varphi_{t_i}(G_{i, j}) - \varphi_{t_i}(G_{i, j + 1})}, \quad q_{i, j}^- := 1- q_{i,j}^+,\quad\mbox{where}\quad G_{i,j} := \hP(X_{t_i} \ge x_{i,j}).
\end{equation}
We assume further that $G_{i, i+1} :=0$ and $\varphi_{ 0}(p) := p$. From (\ref{qij}), in order to have $0< q_{i,j}^+< 1$, it suffices to (and we will) assume that
\bea
\label{mon2}
\varphi_{t_i}(G_{i, j + 1})< \varphi_{t_{i+1}}(G_{i+1, j+1} )  < \varphi_{t_i}(G_{i, j}), \quad \mbox{for all $(i,j)$}.
\eea
Intuitively, \eqref{mon2} is a technical condition which states that $\varphi_{\cdot}$ cannot change too quickly in time. Clearly this condition is satisfied when $\varphi_t \equiv \varphi$. Now let $\hQ$ be the (equivalent) probability measure under which $X$ is Markov with transition probabilities given by
\begin{equation}
\label{Qdiscrete}
\hQ\big(X_{t_{i+1}} = x_{i+1, j+1} \big| X_{t_i} = x_{i,j}\big) = q_{i,j}^+,\quad \hQ\big(X_{t_{i+1}} = x_{i+1, j} \big| X_{t_i} = x_{i,j}\big) = q_{i,j}^-.
\end{equation}

We first have the following simple lemma.

\begin{lemma} \label{lem-ui}
Assume \eqref{mon2} holds and $g\in \cI$.  For $0<n\le N$, define $u_n(x) := g(x)$, and  for $i=n-1,\ldots, 0$, 
\bea
\label{ui}
u_i(x_{i,j}) := q_{i,j}^+ u_{i+1}(x_{i+1, j+1}) + q_{i,j}^- u_{i+1}(x_{i+1,j}),\quad j=0,\ldots, i.
\eea  
Then $u_i$ is increasing and $\hE^\hQ[g(X_{t_n})|\cF_{t_i}] = u_i(X_{t_i})$.
\end{lemma}
\begin{proof}  
It is obvious from the binomial tree structure that $\hE^\hQ[g(X_{t_n})|\cF_{t_i}] = u_i(X_{t_i})$. We prove the monotonicity of $u_i$ by backward induction. First, $u_n = g$ is increasing. Assume $u_{i+1}$ is increasing. Then, noting that $x_{i, j}$'s are increasing in $j$, and $q^+_{i,j}+q^-_{i,j}=1$ for all $i,j$, by
(\ref{ui}) we have 
\beaa
u_i(x_{i,j}) &\le&  q_{i,j}^+ u_{i+1}(x_{i+1, j+1}) + q_{i,j}^- u_{i+1}(x_{i+1,j+1}) = u_{i+1}(x_{i+1,j+1}) \\
&\le& q_{i,j+1}^+ u_{i+1}(x_{i+1, j+2}) + q_{i,j+1}^- u_{i+1}(x_{i+1,j+1}) = u_i(x_{i, j+1}). 
\eeaa
Thus $u_i$ is also increasing.  
\end{proof}

We remark that \eqref{ui} can be viewed as a  ``discrete partial differential equation". This idea motivates our treatment of the continuous time model in the next section. 

The following is our main result of this section.

\begin{theorem} \label{thm-ddf}
	Assume \eqref{mon2}. Then there exists a unique time-consistent dynamic distortion function $\Phi$ such that $\Phi(t_0, t_n, x_{0,0}; p) = \varphi_{t_n}(p)$ for $n=1,\ldots, N$, and for all $0\le i < n \le N$, $0\le j\le i$, and $0\le k \le n$, we have
	\begin{equation}
	\label{Phi}
	\Phi(t_i, t_n, x_{i,j}; \hP\{X_{t_n} \ge x_{n,k}\big|X_{t_i} = x_{i,j}\}) = \hQ\{X_{t_n} \ge x_{n,k}\big|X_{t_i} = x_{i,j}\}.
	\end{equation}
	Here uniqueness is only at the conditional survival probabilities for all $k$ in the left side of \eqref{Phi}. 
	
	Moreover, the corresponding conditional  nonlinear expectation satisfies \eqref{sEQ}.
	
\end{theorem}
\begin{proof}  
We first show that \eqref{Phi} has a solution satisfying the desired initial conditon. Note that  both $ \hP\{X_{t_n} \ge x_{n,k}\big|X_{t_i} = x_{i,j}\}$ and $\hQ\{X_{t_n} \ge x_{n,k}\big|X_{t_i} = x_{i,j}\}$ are strictly decreasing in $k$, for fixed $0\le i < n\le N$ and $x_{i,j}$. Then  one can easily define a function $\Phi$, depending on $t_i, t_n$, $x_{i,j}$, so that \eqref{Phi} holds for all $x_{n,k}$, $0\le k\le n$.  Moreover, the initial condition $\Phi(t_0, t_n, x_{0,0}; p) := \varphi_{t_n}(p)$ is equivalent to
\bea
\label{Qflow}
\varphi_{t_n}(\hP\{X_{t_n} \ge x_{n,k}\}) = \hQ\{X_{t_n} \ge x_{n, k}\},\quad 0\le n \le N, \quad 0\le k\le n.
\eea
We shall prove \eqref{Qflow} by induction on $n$. First recall $\varphi_0(p) = p$ and that $\hP\{X_{t_0} = x_{0,0}\} = \hQ\{X_{t_0} = x_{0,0}\}=1$,  thus \eqref{Qflow}  obviously holds for $n=0$.  Assume now it holds for $n<N$. Then
\beaa
&&\hQ\{X_{t_{n+1}} = x_{n+1, k}\} = \hQ\{X_{t_n} = x_{n,k-1}\} q_{n, k-1}^+ + \hQ\{X_{t_n} = x_{n,k}\}q_{n, k}^-\\
&=&\big[\hQ\{X_{t_n} \ge x_{n,k-1}\}-\hQ\{X_{t_n} \ge  x_{n,k}\}\big] q_{n, k-1}^+ + \big[\hQ\{X_{t_n} \ge x_{n,k}\}- \hQ\{X_{t_n} \ge  x_{n,k+1}\}\big] q_{n, k}^-\\
&=&\big[\varphi_{t_n}(G_{n, k-1}) - \varphi_{t_n}(G_{n,k})\big] q_{n, k-1}^+ + \big[\varphi_{t_n}(G_{n, k}) - \varphi_{t_n}(G_{n,k+1})\big] [1-q_{n, k}^+]\\
&=&\big[\varphi_{t_{n+1}}\big(G_{n+1, k}) - \varphi_{t_n}(G_{n,k})\big]  +  \big[\varphi_{t_n}(G_{n, k}) - \varphi_{t_{n+1}}\big(G_{n+1, k+1} ) \big] \\
&=& \varphi_{t_{n+1}}\big(G_{n+1, k})   - \varphi_{t_{n+1}}\big(G_{n+1, k+1} ).
\eeaa
This leads to \eqref{Qflow} for $n+1$ immediately and thus completes the induction step.

We next show that the above constructed $\Phi$ is indeed a time-consistent dynamic distortion function.  We first remark that, for this discrete model only the values of $\Phi$ in the left side of \eqref{Phi} are relevant, and one may extend $\Phi$ to all $p\in [0,1]$ by linear interpolation. Then by \eqref{Phi} it is straightforward to show that $\Phi(t_i, t_n, x_{i,j};\cdot)$ satisfies Definition \ref{defn-distortion} (i). Moreover, by \eqref{condcE}, \eqref{cEdiscrete} and \eqref{Phi}, for any $g\in \cI$ we have
\beaa
\sE_{t_i, t_n}[g(X_{t_n})] \big|_{X_{t_i} = x_{i,j}} &&= \sum_{k=0}^n g(x_{n,k})  \Big[\Phi\big(t_i, t_n, x_{i,j}; \hP\{X_{t_n} \ge x_{n,k}\big| X_{t_i}=x_{i,j}\}\big) \\
&& \quad \quad - \Phi\big(t_i, t_n, x_{i,j}; \hP\{X_{t_n} \ge x_{n,k+1}\big| X_{t_i}=x_{i,j}\}\big)\Big]\\
&&=  \sum_{k=0}^n g(x_{n,k})  \Big[\hQ\{X_{t_n} \ge x_{n,k}\big| X_{t_i}=x_{i,j}\} - \hQ\{X_{t_n} \ge x_{n,k+1}\big| X_{t_i}=x_{i,j}\}\Big]\\
&&= \hE^\hQ\big[g(X_{t_n})\big|X_{t_i} = x_{i,j}\big].
\eeaa
That is, \eqref{sEQ} holds. Moreover, fix $n$ and $g$ and let $u_i$ be as in Lemma \ref{lem-ui}. Since $u_m$ is increasing, we have
\beaa
\sE_{t_i, t_m}\big[ \sE_{t_m,t_n}[g(X_{t_n})]\big] = \sE_{t_i,t_m} [ u_m(X_{t_m})] = u_i(X_{t_i}) = \sE_{t_i,t_n} [g(X_{t_n})],\quad 0\le i < m < n.
\eeaa
This verifies \eqref{flow2}. Thus $\Phi$ is a time-consistent dynamic distortion function.

It remains to prove the uniqueness of $\Phi$. Assume $\Phi$ is an arbitrary time-consistent dynamic distortion function.  For any appropriate $i$, $j$, and $g\in \cI$, following the arguments of Lemma \ref{lem-ui} we see that
\beaa
u(t_i, x_{i,j}) &:=& \sE_{t_i, t_{i+1}}\big[g(X_{t_{i+1}})]\big|_{X_{t_i} = x_{i, j}} \\
&=&  g(x_{i+1, j}) [1-\Phi(t_i, t_{i+1}, x_{i,j}; p_{i,j}^+)] + g(x_{i+1, j+1}) \Phi(t_i, t_{i+1}, x_{i,j}; p_{i,j}^+)
\eeaa
is increasing in $x_{i,j}$. Then by \eqref{cEdiscrete} and the tower property we have
\beaa
&& \sum_k g(x_{i+1, k}) [\varphi_{t_{i+1}}(G_{i+1, k}) - \varphi_{t_{i+1}}(G_{i+1, k+1})] = \sE_{0, t_{i+1}}[g(X_{t_{i+1}})] = \sE_{0, t_i}\big[\sE_{t_i, t_{i+1}}[g(X_{t_{i+1}}]\big]\\
&&=\sE_{0, t_i}\big[u(t_i, X_{t_{i}})\big] = \sum_j u(t_i, x_{i, j}) [\varphi_{t_{i}}(G_{i, j}) - \varphi_{t_{i}}(G_{i, j+1})] \\
&&=\sum_j  \Big[g(x_{i+1, j}) [1-\Phi(t_i, t_{i+1}, x_{i,j}; p_{i,j}^+)] + g(x_{i+1, j+1}) \Phi(t_i, t_{i+1}, x_{i,j}; p_{i,j}^+)\Big]\times\\
&&\quad\quad [\varphi_{t_{i}}(G_{i, j}) - \varphi_{t_{i}}(G_{i, j+1})].
\eeaa
By the arbitrariness of $g\in \cI$, this implies that 
\beaa
&& 1 - \varphi_{t_{i+1}}(G_{i+1, 1}) = [1-\Phi(t_i, t_{i+1}, x_{i,0}; p_{i,0}^+)] [1 - \varphi_{t_{i}}(G_{i, 1})];\\
&&\varphi_{t_{i+1}}(G_{i+1, k}) - \varphi_s{t_{i+1}}(G_{i+1, k+1}) = [1-\Phi(t_i, t_{i+1}, x_{i,k}; p_{i,k}^+)] [\varphi_{t_{i}}(G_{i, k}) - \varphi_{t_{i}}(G_{i, k+1})]\\
&&\quad\quad + \Phi(t_i, t_{i+1}, x_{i,k-1}; p_{i,k-1}^+) [\varphi_{t_{i}}(G_{i, k-1}) - \varphi_{t_{i}}(G_{i, k})],\quad k=1,\ldots, i+1.
\eeaa
This is equivalent to, denoting $a_k:= \Phi(t_i, t_{i+1}, x_{i,k}; p_{i,k}^+) [\varphi_{t_{i}}(G_{i, k}) - \varphi_{t_{i}}(G_{i, k+1})]$,
\beaa
a_0 &=& \varphi_{t_{i+1}}(G_{i+1, 1}) - \varphi_{t_{i}}(G_{i, 1});\\
a_{k-1} - a_k &=& [\varphi_{t_{i+1}}(G_{i+1, k}) - \varphi_{t_{i+1}}(G_{i+1, k+1}) ] - [\varphi_{t_{i}}(G_{i, k}) - \varphi_{t_{i}}(G_{i, k+1})].
\eeaa
Clearly the above equations have a unique solution, so we must have $\Phi(t_i, t_{i+1}, x_{i,k}; p_{i,k}^+)= q_{i, k}^+$. This implies further that $\sE_{t_i, t_{i+1}}[g(X_{t_{i+1}})] = \hE^\hQ[g(X_{t_{i+1}}|\cF_{t_i}]$. Now both $\sE_{t_i, t_n}$ and $\hE^\hQ[\cdot|\cdot]$ satisfy the tower property, then $\sE_{t_i, t_{n}}[g(X_{t_{n}})] = \hE^\hQ[g(X_{t_{n}}|\cF_{t_i}]$ for all $t_i<t_n$ and all $g\in \cI$. So $\sE_{t_i, t_{n}}$ is unique, which implies immediately the uniqueness of $\Phi$.  
\end{proof}

\medskip

\begin{remark}\label{remark4} { \ }
	\begin{enumerate}
		\item[(i)] We should note that the dynamic distortion function $\Phi$ that we constructed actually depends on the survival function of $X$ under both $\hP$ and $\hQ$, see also \eqref{Phi-cont} below.
		\item[(ii)] Our construction of $\Phi$ is local in time. In particular, all the results can be easily extended to the case with infinite times: $0=t_0<t_1<\cdots$.  
		\item[(iii)] Our construction of $\Phi$ is also local in state, in the sense that $\Phi(t_i, t_n, x_{i, j}; \cdot)$ involves only the subtree rooted at $(t_i, x_{i,j})$.
	\end{enumerate}
\end{remark}

\section{The constant diffusion case.} \label{sec:continuous.time}
In this section we set $\cT = [0, T]$, and consider the case where the underlying state process $X$ is a one dimensional Markov process satisfying the following SDE with constant diffusion coefficient:
\bea
\label{X2}
X_t = x_0 + \int_0^t b(s, X_s) ds +B_t,
\eea
where $B$ is a one-dimensional standard Brownian motion  on a given  filtered probability space $(\Omega, \mathcal{F}, \{\mathcal{F}_t\}_{0 \leq t \leq T}, \mathbb{P})$. Again we are given initial distortion functions $\{\varphi_t = \varphi(t, \cdot) \}_{0< t\le T}$ and $\varphi_0(p) \equiv p$. Our goal is to construct a time-consistent dynamic distortion function $\Phi$ and the corresponding time-consistent distorted conditional  expectations $\sE_{s, t}$ for $(s, t) \in \cT_2$. We shall impose  the following technical conditions. 

\begin{assumption}
\label{assum-b}
The function $b$ is sufficiently smooth and both $b$ and the required derivatives are bounded.
\end{assumption}

Clearly, under the assumption the SDE \eqref{X2} is wellposed. The further regularity of $b$ is used to derive some tail estimates for the density of $X_t$, which are required for our construction of the time-consistent dynamic distortion function $\Phi$ and the distorted probability measure $\hQ$. By investigating our arguments more carefully, we can figure out the precise technical conditions we will need. However, since our main focus is the dynamic distortion function $\Phi$, we prefer not to carry out these details for the sake of the readability of the paper. 

\subsection{Binomial tree approximation.} \label{sec:approx}
Our idea is to approximate $X$ by a sequence of binomial trees and then apply the results from the previous section.  To this end, for fixed $N$, denote $h := {T\over N}$, and $t_i := ih$, $i=0,\cdots, N$.  Then \eqref{X2} may be discretized as follows:
\bea
\label{Euler}
X_{t_{i+1}} \approx X_{t_i} + b(t_i, X_{t_i}) h + B_{t_{i+1}} - B_{t_i}.
\eea
We first construct the binomial tree on $\cT_N:= \{t_i, i=0,\ldots, N\}$ as in Subsection \ref{sect:GBT} with 
\begin{equation}
\label{BTN}
x_{0,0} =x_0,\quad x_{i,j} = x_0 + (2j-i) \sqrt{h},\quad  b_{i,j} := b(t_i, x_{i,j}),\quad p_{i,j}^+ := {1\over 2} + {1\over 2} b_{i,j} \sqrt{h}.
\end{equation}
Since $b$ is bounded, we shall assume $h$ is small enough so that $0< p_{i,j}^+ < 1$. Let $X^N$ denote the Markov chain corresponding to this binomial tree under the probability $\hP_N$ specified by \eqref{BTN}. Then our choice of $p_{i,j}^+$ ensures that
\bea
\label{BTNPN}
&\hE^{\hP_N}\big[X^N_{t_{i+1}} -  X^N_{t_i}  \big| X^N_{t_i} = x_{i,j}\big] =  p_{i, j}^+ \sqrt{h} - p_{i,j}^- \sqrt{h} =  b_{i,j}h;\\
&\hE^{\hP_N}\big[ \big(X^N_{t_{i+1}} -  X^N_{t_i}  - b_{i,j} h\big)^2 \big| X^N_{t_i} = x_{i,j}\big] =  p_{i, j}^+ ( \sqrt{h} - b_{i,j} h)^2+ p_{i,j}^- (\sqrt{h} + b_{i,j}h)^2=  h - b^2_{i,j}h^2.&\nonumber
\eea 
Clearly, as a standard Euler approximation, $X^N$  matches the conditional expectation and conditional variance of $X$ in \eqref{Euler}, up to terms of order $o(h)$.

Next we define the other terms in Section \ref{sect:GBT}:
\bea
\label{GQPhiN}
\left\{\ba{lll}
G^N_{i,j} := \hP_N\{X^N_{t_i} \ge x_{i,j}), \quad q_{i, j}^{N,+} :=  \frac{\varphi_{t_{i+1}}(G^N_{i+1, j+1} ) - \varphi_{t_i}(G^N_{i, j + 1})}{\varphi_{t_i}(G^N_{i, j}) - \varphi_{t_i}(G^N_{i, j + 1})},\quad q_{i, j}^{N,-} := 1- q_{i,j}^{N,+}; \\
\hQ_N\big\{X^N_{t_{i+1}} = x_{i+1, j+1} \big| X^N_{t_i} = x_{i,j}\big\} = q_{i,j}^{N,+},\quad \hQ_N\big\{X^N_{t_{i+1}} = x_{i+1, j} \big| X^N_{t_i} = x_{i,j}\big\} = q_{i,j}^{N,-}; \\
\Phi_N\big(t_i, t_n, x_{i,j}; \hP_N\big\{X^N_{t_n} \ge x_{n,k}\big|X^N_{t_i} = x_{i,j}\big\}\big) :=\hQ_N\big\{X^N_{t_n} \ge x_{n,k}\big|X^N_{t_i} = x_{i,j}\big\}.
\ea\right.
\eea

We shall send $N\to\infty$ and analyze the limits of the above terms. In this subsection we evaluate the limits heuristically, by assuming  all  functions involved exist and are smooth.  

Define the survival probability function  and density  function of the $X$ in \eqref{X2}, respectively:
\bea
\label{Grho}
G(t,x) := \hP(X_t \ge x),\quad \rho(t,x) := - \pa_x G(t,x),\quad 0<t \le T.
\eea
Note that, as the survival function of the diffusion process \eqref{X2}, $G$ satisfies the following PDE:
\bea
\label{heat}
\pa_t G = {1\over 2} \pa_{xx} G -  b \pa_x G = -{1\over 2} \pa_x \rho + b \rho.
\eea

It is reasonable to assume $G^N_{i,j} \approx G(t_i, x_{i,j})$. Note that, $t_{i+1} = t_i + h$, $x_{i, j+1} = x_{i, j} + 2\sqrt{h}$, $x_{i+1, j+1} = x_{i,j} + \sqrt{h}$.  Rewrite $\varphi(t,p) := \varphi_t(p)$.  Then, for $(t,x) = (t_i, x_{i,j})$, by \eqref{heat} and applying Taylor expansion  we have (suppressing variables 
when the context is clear):
\beaa
&& \varphi\big(t+h, G(t+h, x+ \sqrt{h})\big) - \varphi(t, G(t,x))  \\
&&= \pa_t \varphi h + \pa_p \varphi [\pa_t G h + \pa_x G \sqrt{h} + {1\over 2} \pa_{xx} G h] + {1\over 2} \pa_{pp} \varphi [\pa_x G]^2 h + o(h)\\
&&=  - \pa_p \varphi \rho \sqrt{h} + \big[\pa_t \varphi + \pa_p \varphi b \rho - \pa_p \varphi \pa_x \rho + {1\over 2} \pa_{pp} \varphi \rho^2\big] h + o(h);\\
&& \varphi\big(t, G(t, x+2\sqrt{h})\big)  - \varphi(t, G(t,x))  = \pa_p \varphi[\pa_x G 2\sqrt{h} + {1\over 2} \pa_{xx} G 4 h] + {1\over 2} \pa_{pp}  [\pa_x G]^24 h + o(h)\\
&&= - 2 \pa_p \varphi  \rho \sqrt{h} - 2\big[ \pa_p\varphi \pa_{x} \rho - \pa_{pp}\varphi \rho^2\big] h + o(h) .
\eeaa  
Thus we have an approximation for the $q^{N,+}_{i,j}$ in \eqref{GQPhiN}:  
\bea
\label{qN+}
q_{i, j}^{N,+} &\approx&  \frac{\varphi\big(t+h, G(t+h, x+ \sqrt{h})\big) - \varphi\big(t, G(t, x+ 2\sqrt{h}\big)}{\varphi\big(t, G(t,x)\big) - \varphi\big(t, G(t, x+2\sqrt{h})\big)}\nonumber\\
&=&1+ {- \pa_p \varphi \rho \sqrt{h} + \big[\pa_t \varphi + \pa_p \varphi b \rho - \pa_p \varphi \pa_x \rho + {1\over 2} \pa_{pp} \varphi \rho^2\big]h + o(h) \over  2 \pa_p \varphi \rho \sqrt{h} + 2\Big[\pa_p \varphi \pa_x \rho -  \pa_{pp}\varphi \rho^2\Big] h + o(h)}\\
&=&  {1\over 2} + {1\over 2} \mu(t,x) \sqrt{h} + o(\sqrt{h}), \nonumber
\eea 
where
\bea
\label{mu}
\mu(t,x) :=  b(t,x) + {\pa_t \varphi (t, G(t,x))  - {1\over 2} \pa_{pp}\varphi(t, G(t,x)) \rho^2(t,x) \over \pa_p \varphi(t, G(t,x)) \rho(t,x)}.
\eea

Next, note that
\beaa
&\hE^{\hQ_N}\big\{X^N_{t_{i+1}} - X^N_{t_i} \big| X^N_{t_i}=x_{i,j}\big\} = \sqrt{h} [2 q^{N,+}_{i,j} - 1]  = \mu(t_i, x_{i,j}) h + o(h);\\
&\hE^{\hQ_N}\big\{(X^N_{t_{i+1}} - X^N_{t_i}  -  \mu(t_i, x_{i,j}) h)^2 \big| X^N_{t_i}=x_{i,j}\big\} = h + o(h).
\eeaa
In other words, as $N\to\infty$, we expect that $\hQ_N$ would converge  to a probability measure $\hQ$,  such that  for some $\hQ$-Brownian motion $\tilde B$, it holds that
\bea
\label{XQ}
X_t = x_0 + \int_0^t \mu(s, X_s) ds + \tilde B_t,\quad \hQ\mbox{-a.s.}
\eea
Moreover, formally one should be able to find a  dynamic distortion function $\Phi$ satisfying: 
\begin{equation}
\label{PhiX}
\Phi\big(s,t, x; \hP\{X_t \ge y | X_s = x\}\big) = \hQ\{X_t \ge y | X_s = x\}, \quad 0\le s< t \le T.
\end{equation}
We shall note that, however, since $X_0=x_0$ is degenerate, $\rho(0,\cdot)$ and hence $\mu(0,\cdot)$ do not exist, so the above convergence will hold only for $0<s<t\le T$. It is also worth noting that asymptotically \eqref{ui} should read:
\beaa
u(t, x) &\approx& {1\over 2} [1+  \mu(t,x) \sqrt{h} + o(\sqrt{h})] [u(t+h, x+ \sqrt{h}) - u(t+h, x-\sqrt{h}) ] + u(t+h, x-\sqrt{h})\\
&=& u(t,x) +\left[\pa_t u + {1\over 2} \pa_{xx} u  + \mu \pa_x u\right] h + o(h).
\eeaa  
That is, 
\bea
\label{cLu}
\sL u(t,x) := \pa_t u + {1\over 2} \pa_{xx} u  + \mu \pa_x u = 0.
\eea

\subsection{Rigorous results for the continuous time model.}
\label{sect:limit}
We now substantiate the heuristic arguments in the previous subsection and derive the time-consistent dynamic distortion function and the distorted conditional expectation for the continuous time model.  We first have the following tail estimates for the density of the diffusion \eqref{X2}. Since our main focus is the dynamic distortion function we postpone the proof to Section \ref{sect-density} below. 

\begin{proposition} \label{prop-tail}
Under Assumption \ref{assum-b}, $X_t$ has a density function $\rho(t,x)$ which is  strictly positive and sufficiently smooth on $(0, T]\times \hR$. Moreover, for any $0<t_0\le T$, there exists a constant $C_0$, possibly depending on $t_0$, such that
\begin{equation}
\label{Gbound}
{|\pa_x \rho(t,x)|\over \rho(t,x)} \le C_0,\quad  {1\over C_0[1+|x|]} \le {G(t,x) [1-G(t,x)]\over \rho(t,x)} \le C_0,\quad (t,x)\in [t_0, T]\times \hR.
\end{equation}
\end{proposition}

We next assume the following technical conditions on $\varphi$.

\begin{assumption} \label{assum-Gphi}
	$\varphi$ is continuous on $[0, T]\times [0, 1]$ and is sufficiently smooth in $(0, T]\times (0,1)$ with $\pa_p \varphi >0$. Moreover, for any $0< t_0 < T$, there exists a constant $C_0>0$ such that for $(t,p) \in [t_0, T]\times (0,1)$ we have the following bounds:
	\bea
	\label{phibound}
	 \Big| {\pa_{pp}\varphi(t,p) \over \pa_p \varphi(t,p)}\Big|\le {C_0\over p(1-p)},\quad  \Big| {\pa_{ppp} \varphi(t,p) \over \pa_p \varphi(t,p)}\Big|\le {C_0\over p^2(1-p)^2},\\
	 \Big|{\pa_t \varphi (t,p) \over \pa_p \varphi(t,p)}\Big| \le C_0 p(1-p),\quad \Big|{\pa_{tp} \varphi(t,p) \over \pa_p \varphi(t,p)}\Big| \le C_0.
	\eea
\end{assumption}

We note that, given the existence of $G(t, x)$, $\rho(t, x)$ as well as the regularity of $\varphi$,  the function $\mu(t, x)$ in \eqref{mu} is well defined. 

\medskip

\begin{remark} \label{rem-G} { \ }
	\begin{enumerate}
		\item[(i)] Note that in \eqref{qN+} and \eqref{mu} only  the composition $\varphi(t, G(t,x))$ is used, and obviously $0< G(t,x) <1$ for all 
		$(t,x)\in (0, T]\times \hR$. Therefore we do not require the differentiability of $\varphi$ at $p=0, 1$. Moreover, since $\pa_p \varphi >0$, the condition \eqref{phibound} involves only the singularities around $p\approx 0$ and $p\approx 1$.
		\item[(ii)] The first line in \eqref{phibound}  is not restrictive. For example, by straightforward calculation one can verify that all the following distortion functions commonly used in the literature (see, e.g., Huang-NguyenHuu-Zhou \cite[Section 4.2]{HNZ}) satisfy it: recalling that in the literature typically $\varphi(t, p) = \varphi(p)$ does not depend on $t$,
		\begin{itemize}
			\item Tversky and Kahneman \cite{KT2}: $\varphi(p) = {p^{\g} \over (p^{\g} + (1-p)^{\g})^{1 / \g}}$,  $\g\in [\g_0, 1)$, where $\g_0 \approx 0.279$ so that $\varphi$ is increasing.
			\item Tversky and Fox \cite{TF}: $\f(p) = {\a p^\g \over \a p^\g + (1-p)^\g}$, $\a>0, \g\in (0, 1)$.
			\item Prelec \cite{Prelec}: $\f(p) = \exp(-\g(-\ln p)^\a)$, $\g>0$, $\a\in (0,1)$. 
			\item Wang \cite{Wang}: $\f(p) = F(F^{-1}(p) + \a)$, $\a\in \hR$, where $F$ is the cdf of the standard normal.
		\end{itemize}
	    As an example, we check the last one which is less trivial. Set $q:= F^{-1}(p)$. Then 
	    \beaa
	    \f(F(q)) = F(q+\a) \Longrightarrow \f'(F(q)) = {F'(q+\a)\over F'(q)} \Longrightarrow \ln(\f'(F(q)) ) = \ln(F'(q+\a)) - \ln(F'(q)). 
	    \eeaa
	    Note that $F'(q) = {1\over \sqrt{2\pi}} e^{-{q^2\over 2}}$, then $\ln (F'(q)) = -\ln\sqrt{2\pi} - {q^2\over 2}$. Thus
	    \beaa
	    \ln(\f'(F(q)) ) = -{(q+\a)^2\over 2} + {q^2\over 2} ~\Longrightarrow~ {\f''(F(q))\over \f'(F(q))} F'(q) = -\a. 
	    \eeaa
	    This implies that, denoting  by $G(q) := 1-F(q)$  the survival function of the standard normal,
	    \beaa
	    {|\f''(p)|\over \f'(p)} p[1-p] = |\a| {F(q)[1-F(q)] \over F'(q)} = |\a| {G(q)[1-G(q)] \over F'(q)},
	    \eeaa
	    Then by applying \eqref{Gbound} on standard normal (namely $b=0$ and $t=1$ there) we obtain the desired estimate for ${\pa_{pp}\f\over \pa_p \f}$. Similarly we may estimate ${\pa_{ppp}\f\over \pa_p \f}$.
	    \item[(iii)] When  $\f(t, p) \equiv \f(p)$ as in the standard literature, the second line in \eqref{phibound} is trivial. Another important example is the separable case: $\f(t,p) = f(t)\f_0(p)$. Assume $f'$ is bounded. Then the second inequality here becomes trivial, and a sufficient condition for the first inequality is ${\f_0(p)\over \f_0'(p)} \le {C_0\over p(1-p)}$,  which holds true for all the examples in (ii). 
	\end{enumerate}
\end{remark}

\medskip

To have a better understanding about $\mu$ given by \eqref{mu}, we compute an example explicitly.

\medskip

\begin{example} \label{eg-mu}
Consider Wang \cite{Wang}'s distortion function: $\f(t, p) = F(F^{-1}(p)+\a)$, as  in Remark \ref{rem-G}(ii). Set $b=0$, then $\mu(t, x) = {\a\over 2\sqrt{t}}$.
\end{example}
\begin{proof}  
First it is clear that $\pa_t \f =0$ and  $\pa_p \f(t,p) = {F'(F^{-1}(p)+\a)\over F'(F^{-1}(p))}$. Then 
\beaa
{\pa_{pp}\f(t,p)\over \pa_p\f(t,p)} = \pa_p \big[\ln \big(\pa_p \f(t,p)\big)\big] = {1\over F'(F^{-1}(p))}\Big[{F''(F^{-1}(p)+\a))\over F'(F^{-1}(p)+\a)) } - {F''(F^{-1}(p))\over F'(F^{-1}(p))}\Big] .
\eeaa
One can easily check that $F'(x) = {1\over \sqrt{2\pi}} e^{-{x^2\over 2}}$ and $F''(x) = - x F'(x)$.	Then
\beaa
{ \pa_{pp}\f(t,p) \over \pa_p \f(t,p)}={1\over F'(F^{-1}(p))}\big[-[F^{-1}(p)+\a] +F^{-1}(p)\big] = -{\a\over F'(F^{-1}(p))}.
\eeaa
Note that
\beaa
G(t,x) = \hP(B_t \ge x) = \hP(B_1 \ge {x\over \sqrt{t}}) =  \hP(B_1 \le -{x\over \sqrt{t}}) = F(-{x\over \sqrt{t}}).
\eeaa
Then
\beaa
\mu(t,x) = {\a \rho(t,x) \over F'(F^{-1}(G(t,x)))} = {\a \rho(t,x) \over F'(-{x\over \sqrt{t}})} = {{\a\over \sqrt{2\pi t}} e^{-{x^2\over 2t}}\over {1\over \sqrt{2\pi }} e^{-{(-{x\over \sqrt{t}})^2\over 2}}}={\a\over \sqrt{t}},
\eeaa
completing the proof.
\end{proof}

We now give some technical preparations.  Throughout the paper we shall use $C$ to denote a generic constant which may vary from line to line.

\begin{lemma} \label{lem-mu}
Let Assumptions \ref{assum-b} and \ref{assum-Gphi} hold. 
\begin{enumerate}
	\item[(i)] The function $\mu$ defined by \eqref{mu} is sufficiently smooth in $(0, T]\times \hR$. Moreover, for any $0<t_0 <T$, there exists $C_0>0$ such that 
	\begin{equation}
	\label{mugrowth}
	|\mu(t,x)| \le C_0[1+|x|],\quad |\pa_x \mu(t,x)|\le C_0[1+|x|^2], \quad \mbox{for all}~ (t,x) \in [t_0, T]\times \hR.
	\end{equation}
	\item[(ii)] For any $(s, x)\in (0, T)\times \hR$, the following SDE on $[s, T]$ has a unique strong solution:
	\begin{equation}
	\label{tildeXsx}
	\tilde X^{s,x}_t = x + \int_s^t \mu(r, \tilde X^{s,x}_r) dr + B^s_t, \quad \mbox{where}\quad B^s_t:= B_t-B_s, \quad t\in [s, T],~ \hP\mbox{-a.s.}, 
	\end{equation}
	Moreover, the following $M^{s,x}$ is a true $\hP$-martingale and $\hP\circ (\tilde X^{s,x})^{-1} = \hQ^{s,x}\circ (X^{s,x})^{-1}$, where
	\bea
	\label{Q}
	X^{s,x}_t:= x +B^s_t,\quad  M^{s,x}_t := e^{\int_s^t \mu(r, X^{s,x}_r) dB_r -{1\over 2}\int_s^t |\mu(r, X^{s,x}_r)|^2dr},\quad {d\hQ^{s,x}\over d\hP} := M^{s,x}_T.
	\eea
	\item[(iii)] Recall the process $X$ as in  \eqref{X2}. Define
	\begin{equation}
	\label{G}
	G^{s,x}_t(y) := \hP(X_t \ge y | X_s = x),\quad \tilde G^{s,x}_t (y):= \hP(\tilde X^{s,x}_t \ge y ),~ 0<s< t\le T,~ x, y\in \hR.
	\end{equation}
	Then $G^{s,x}_t$ and $\tilde G^{s,x}_t$ are  continuous, strictly decreasing in $y$, and enjoy the following properties:
	\beaa
	\left.\ba{c}
	G^{s,x}_t(\infty) := \lim_{y\to \infty} G^{s,x}_t(y) = 0, \quad \tilde G^{s,x}_t(\infty) := \lim_{y\to \infty} \tilde G^{s,x}_t(y)  = 0;\\
	G^{s,x}_t(-\infty) := \lim_{y\to -\infty} G^{s,x}_t(y) = 1, \quad \tilde G^{s,x}_t(-\infty) := \lim_{y\to -\infty} \tilde G^{s,x}_t(y)  = 1.
	\ea\right.
	\eeaa
	Furthermore,  $G^{s,x}_t$ has a continuous inverse function $(G^{s,x}_t)^{-1}$ on $(0, 1)$, and by continuity we set $(G^{s,x}_t)^{-1}(0) := -\infty$, $(G^{s,x}_t)^{-1}(1) := \infty$.
	
	\item[(iv)] For any $g\in \cI$ fixed, let $u(t,x) := \hE^{\hP }[g(\tilde X^{t,x}_T)]$, $(t,x)\in (0, T]\times \hR$. 
	Then  $u$ is bounded, increasing in $x$, and is the unique bounded viscosity solution of the following PDE:
	\begin{equation}
	\label{PDE}
	\sL u(t,x) := \pa_t u + {1\over 2} \pa_{xx} u  + \mu \pa_x u =0,\quad 0< t\le T;\quad u(T,x) = g(x).
	\end{equation}
	
	\item[(v)] For the $t_0$ and $C_0$ in (i),   there exists $\delta = \delta(C_0)>0$ such that, if  $g\in \cI$ is sufficiently smooth and $g'$ has compact support, then $u$ is sufficiently smooth on $[T-\delta, T]\times \hR$ and there exists a constant  $C>0$, which may depend on $g$, satisfying, for $(t,x)\in [T-\delta, T]\times \hR$, 
	\begin{equation}
	\label{paxubound}
	|u(t,x) - g(-\infty)| \le C e^{-x^2}, x<0;~  |u(t,x) - g(\infty)|\le C e^{-x^2},x>0; ~  \pa_x u(t,x) \le Ce^{- x^2}.
	\end{equation}
\end{enumerate}

\end{lemma}
\begin{proof}  
(i)  By our assumptions and Proposition \ref{prop-tail}, the regularity of $\mu$ follows immediately.  For any $t\ge t_0$, by \eqref{phibound} and then \eqref{Gbound} we have, 
\beaa
&&  \Big|{\pa_t \f(t, G(t,x))\over \pa_p \f(t, G(t,x)) \rho(t,x)}\Big| \le {C G(t,x) [1-G(t,x)]\over \rho(t,x)} \le C;\\
&&  \Big|{\pa_{pp}\f(t, G(t,x)) \rho(t,x) \over \pa_p \f(t, G(t,x)) }\Big| \le {C \rho(t,x) \over G(t,x)[1-G(t,x)]}   \le C  [1+|x|].
\eeaa
Then it follows from \eqref{mu} that  $ |\mu(t,x)| \le C[1+|x|]$.

Moreover, note that
\beaa
\pa_x \mu(t, x) = \pa_x b  - {\pa_{tp} \f \over \pa_p\f} + {\pa_t \f \pa_{pp}\f\over (\pa_p \f)^2} - {\pa_t \f \pa_x \rho\over \pa_p\f \rho^2} + {1\over 2}{\pa_{ppp}\f \rho^2\over \pa_p\f} - {1\over 2}{\pa_{pp}\f \pa_x \rho \over \pa_p\f} - {1\over 2}{(\pa_{pp}\f)^2 \rho^2 \over (\pa_p\f)^2}.
\eeaa
By \eqref{phibound} and  \eqref{Gbound} again one can easily verify that $|\pa_x \mu(t,x)|\le C_0[1+|x|^2]$.

(ii) Since $\mu$ is locally uniform Lipschitz continuous in $x$, by a truncation argument $\tilde X^{s,x}$ exists locally. Now the uniform linear growth \eqref{mugrowth} guarantees the global existence. Moreover,  by Karatzas-Shreve \cite[Chapter 3, Corollary 5.16]{KS} we see that $M^{t,x}$ is a true $\hP$-martingale and thus $\hQ^{t,x}$ is a probability measure. 

(iii) Since the conditional law of $X_t$ under $\hP$ given $X_s=x$ has a  strictly positive density, the statements concerning $G^{s,x}_t$ is obvious. Similarly, since the law of $X^{s,x}_t$ under $\hP$ has a density, and $d\hQ^{s,x} \ll d\hP$, the statements concerning $\tilde G^{s,x}_t$ are also obvious.

(v) We shall prove (v) before (iv). Let $\delta >0$ be specified later, and $t\in [T-\delta, T]$. Let $R>0$ be such that $g'(x)=0$ for $|x|\ge R$. Note that $u(t,x) = \hE[M_T^{t,x}g(X_T^{t,x})]$. For $x > 2R$, we have
\beaa
&&\big|u(t,x) - g(\infty)\big| = \Big|\hE\big[M_T^{t,x}[g(X_T^{t,x}) - g(\infty)]\big]\Big|\le \hE\big[M_T^{t,x}|g(X_T^{t,x}) - g(\infty)|\big]\\
&& \le  2 \| g\|_{\infty} \hE\big[M_T^{t,x}1_{\{X^{t,x}_T \le R\}}\big] =   2 \| g\|_{\infty} \hE\Big[e^{\int_t^T \mu dB_r - {3\over 2}\int_t^T |\mu|^2dr} e^{\int_t^T |\mu|^2dr} 1_{\{X^{t,x}_T \le R\}}\big].
\eeaa
Then
\beaa
\big|u(t,x) - g(\infty)\big|^3 \le  \| g\|_{\infty}^3 \hE\Big[e^{3\int_t^T \mu dB_r - {9\over 2}\int_t^T |\mu|^2dr}\Big]\hE\Big[e^{3\int_t^T |\mu|^2dr}\Big] \hP(X^{t,x}_T \le R).
\eeaa
By \cite[Chapter 3, Corollary 5.16]{KS} again, we have $\hE\Big[e^{3\int_t^T \mu dB_r - {9\over 2}\int_t^T |\mu|^2dr}\Big]=1$. By \eqref{mugrowth}, we have
\beaa
\hE\Big[e^{3\int_t^T |\mu|^2dr}\Big] \le \hE\Big[e^{ C_0 \int_t^T [1+|x|^2 + |B^t_r|^2]dr}\Big]\le e^{ C_0 \delta [1+|x|^2]} \hE\Big[e^{ C_0 \delta \sup_{0\le s\le \delta} |B_s|^2}\Big]\le e^{2 C_0 \delta [1+|x|^2]},
\eeaa
for $\delta$ small enough.  Fix such a $\delta > 0$ and let  $t \in [T - \delta, T]$. Note that we may choose $\delta$ independent from $g$. Henceforth we let $C > 0$ be a generic constant. Moreover, since $x>2R$,
\beaa
\hP(X^{t,x}_T \le R) \le \hP(X^{t,x}_T \le {x\over 2}) = \hP(B^t_T \le -{x\over 2}) \le  \hP(B_1 \le -{x\over 2\sqrt{\delta}}) \le Ce^{-{x^2\over 8\delta}}.
\eeaa
Putting together we have, for $\delta$ small enough,
\beaa
\big|u(t,x) - g(\infty)\big|^3  \le Ce^{-{x^2\over 8\delta} + C\delta [1+|x|^2]} \le Ce^{-3 x^2},
\eeaa
This implies that $\big|u(t,x) - g(\infty)\big| \le Ce^{-x^2}$. Similarly, $\big|u(t,x) - g(-\infty)\big| \le Ce^{-x^2}$ for $x<0$.

The preceding estimates allows us to differentiate inside the expectation, and we have
\beaa
&\pa_x u(t,x) =\hE\big[g'(\tilde X_T^{t,x})\td \tilde X_T^{t,x}\big], \\
&\mbox{where} \quad \td \tilde X_s^{t,x} = 1 + \int_t^s \pa_x\mu(r, \tilde X_r^{t,x}) \td \tilde X_r^{t,x} dr\quad \mbox{and thus}\quad \td \tilde X_T^{t,x} = e^{\int_t^T \pa_x\mu(s, \tilde X_s^{t,x}) ds}>0.
\eeaa
Then $\pa_x u \ge 0$.  Moreover, recalling \eqref{Q} we  have
\beaa
\pa_x u(t,x) &=&  \hE\Big[g'(\tilde X_T^{t,x})e^{\int_t^T \pa_x\mu(s, \tilde X_s^{t,x}) ds}\Big]=\hE\Big[M^{t,x}_Tg'(\tilde X_T^{t,x})e^{\int_t^T \pa_x\mu(s, \tilde X_s^{t,x}) ds}\Big] \\
&\le& C \hE\Big[M^{t,x}_T e^{\int_t^T \pa_x\mu(s, \tilde X_s^{t,x}) ds} 1_{\{|X_T^{t,x}|\le R\}}\Big].
\eeaa
Then, by the estimate of $\pa_x\mu$ in \eqref{mugrowth}, it follows from the same arguments as above we can show that $|\pa_x u(t,x)|\le Ce^{-x^2}$. 

Finally, we may apply the arguments further to show that $u$ is sufficiently smooth, and then it follows from the flow property and the standard It\^{o} formula that $u$ satisfies PDE \eqref{PDE}. 

(iv) We shall only prove the results on $[T-\delta, T]$. Since $\delta>0$ depends only on $C_0$ in (i), one may apply the results backwardly in time and extend the results to $[t_0, T]$. Then it follows from the arbitrariness of $t_0$ that the results hold true on $(0, T]$.

We now fix $\delta$ as in (v). The boundedness of $u$ is obvious. Note that $u(t,x) =  \hE^{\hP}\big[g(X^{t,x}_{T}) M^{t,x}_T\big]$, and $\mu$, $g$ are continuous, following similar arguments as in (v) one can show that $u$ is continuous. 

Next, for any $g\in \cI$, there exist approximating sequence $\{g_n\}$ such that each $g_n$ satisfies the conditions in (v). Let $u_n(t,x) := \hE[g_n(\tilde X^{t,x}_T)]$. Then $u_n$ is increasing in $x$ and is a classical solution to PDE \eqref{PDE} on $[T-\delta, T]$ with terminal condition $g_n$. It is clear that $u_n \to u$. Then $u$ is also increasing in $x$ and its viscosity property follows from the stability of viscosity solutions. The uniqueness of viscosity solution follows from the standard comparison principle. We refer to the classical reference Crandall-Ishii-Lions \cite{CIL} for the details of the viscosity theory.  
\end{proof}

\medskip

We are now ready for the main result of this section. Recall \eqref{G} and define
\bea
\label{Phi-cont}
\Phi(s,t,x; p) := \tilde G^{s,x}_t\big((G^{s,x}_t)^{-1}(p)\big), \quad s>0.
\eea

\begin{theorem}
	\label{thm-limit}
	Let Assumptions \ref{assum-b} and \ref{assum-Gphi} hold.
	Then $\Phi$ defined by \eqref{Phi-cont} is a time-consistent dynamic distortion function which is consistent with the initial conditions: $\Phi(0,t, x; p) = \f_t(p)$. 
\end{theorem}
\begin{proof}  
First, by Lemma \ref{lem-mu}(iv) it is straightforward to check that $\Phi$ satisfies Definition \ref{defn-ddf}(i).

Next, For $0<s<t\le T$, note that the definition \eqref{Phi-cont} of $\Phi$ implies the counterpart of \eqref{Phi}:
\bea
\label{PhiQ}
\Phi(s,t,x; G^{s,x}_t(y)) = \tilde G^{s,x}_t(y).
\eea
Recall \eqref{condcE} and Lemma \ref{lem-mu}, one can easily see that $\sE_{s, t}[g(X_t)] = u(s, X_s)$ for any $g\in \cI$, where $u(s, x) := \hE^\hP[g(\tilde X^{s,x}_t)]$ is  increasing in $x$ and is the unique viscosity solution of the PDE \eqref{PDE} on $[s, t]\times \hR$ with terminal condition $u(t,x) = g(x)$. Then, either by the flow property of  the solution to SDE \eqref{tildeXsx} or the uniqueness of the  PDE, we obtain the tower property \eqref{flow2} immediately for $0<r<s<t\le T$. 

To verify the tower property at $r=0$, let $t_0>0$ and $\delta >0$ be as in Lemma \ref{lem-mu} (v). We first show that, for  any $g$ as in Lemma \ref{lem-mu}(v) and the  corresponding $u$, we have 
\bea
\label{dsE}
\sE_{0,t_1}\big[u(t_1, X_{t_1})]\big]=\sE_{0,t_2}\big[u(t_2, X_{t_2})]\big],\quad T-\delta \le t_1 < t_2 \le T.
\eea
Clearly the set of such $g$ is dense in $\cI$, then \eqref{dsE} holds true for all $g\in \cI$, where $u$ is the viscosity solution to the PDE \eqref{PDE}. Note that $u(t, X_t) = \sE_{t, T}[g(X_T)]$, then by setting $t_1=t$ and $t_2=T$ in \eqref{dsE} we obtain $ \sE_{0,t}\big[\sE_{t,T}[g(X_T)]\big] =\sE_{0,T}[g(X_T)]  $ for $T-\delta \le t\le T$. Similarly we can verify the tower property over any interval $[t-\delta, t] \subset [t_0, T]$. Since $\sE_{s, t}$ is already time-consistent for $0<s<t$, we see the time-consistency for any $t_0\le s<t\le T$. Now by the arbitrariness of $t_0>0$, we obtain the tower property at $r=0$ for all $0<s<t\le T$.

We now prove \eqref{dsE}. Recall \eqref{paxubound} that $u(t,-\infty) = g(-\infty)=0$.  
Then, for $T-\delta \le t\le T$,  similar to (\ref{sEgsmooth}), we have
\beaa
\sE_{0,t}\big[u(t, X_t)]\big] = \int_0^\infty \f\big(t, \hP(u(t, X_t) \ge x)\big) dx = \int_\hR \f\big(t, G(t,x)\big) \pa_x u(t,x) dx.
\eeaa
Let $\psi_m: \hR\to [0, 1]$ be smooth with $\psi_m(x) = 1$, $|x|\le m$, and $\psi_m(x) =0$, $|x|\ge m+1$.
Denote
\beaa
\sE^m_{0,t}\big[u(t, X_t)]\big] :=  \int_\hR \f\big(t, G(t,x)\big) \pa_x u(t,x) \psi_m(x)dx.
\eeaa

Then, recalling \eqref{heat}, \eqref{PDE}, and suppressing the variables when the context is clear, we have
\bea
\label{dsEt}
&&{d\over dt} \sE^m_{0,t}\big[u(t, X_t)]\big] = \int_\hR \Big[[\pa_t \f + \pa_p \f \pa_t G] \pa_x u + \f \pa_{tx} u\Big]\psi_m dx\nonumber\\ 
&& =  \int_\hR \Big[[\pa_t \f + \pa_p \f \pa_t G] \pa_x u \psi_m + [\pa_p \f  \rho \psi_m- \f \psi_m']\pa_{t} u\Big] dx\nonumber\\
&&= \int_\hR \Big[\big[\pa_t \f + \pa_p \f [b \rho - {1\over 2}\pa_x \rho]\big] \psi_m\pa_x u  - [\pa_p \f  \rho\psi_m  - \f\psi_m']  \big[  {1\over 2} \pa_{xx} u + \mu \pa_x u\big] \Big] dx\nonumber\\
&&= \int_\hR \Big[\big[\pa_t \f + \pa_p \f [b \rho - {1\over 2}\pa_x \rho]- \pa_p \f  \rho \mu\big]\psi_m \pa_x u  + \f\psi_m'  \mu \pa_x u \\
&&\quad +{1\over 2} \pa_{x} u \big[\pa_p \f \pa_x \rho \psi_m- \pa_{pp} \f \rho^2\psi_m +2\pa_p \f  \rho\psi'_m - \f \psi_m''  \big] \Big] dx\nonumber\\
&&=\int_\hR \Big[\big[\pa_t \f + \pa_p \f b \rho - \pa_p \f  \rho \mu - {1\over 2} \pa_{pp} \f \rho^2\big]\psi_m + \big[\f \mu + \pa_p \f \rho\big]\psi'_m  - {1\over 2} \f \psi''_m \Big]\pa_x u   dx\nonumber\\
&&=  \int_\hR  \Big[ \big[\f \mu + \pa_p \f \rho\big]\psi'_m  - {1\over 2} \f \psi''_m \Big]\pa_x u dx,\nonumber
\eea

where the last equality follows from \eqref{mu}. That is,  for any $T-\delta\le t_1 < t_2 \le T$,
\begin{equation}
\label{dsEt2}
\sE^m_{0,t_2}\big[u(t_2, X_{t_2})]\big]- \sE^m_{0,t_1}\big[u(t_1, X_{t_1})]\big] = \int_{t_1}^{t_2}\int_\hR  \Big[\big[\f \mu + \pa_p \f \rho\big]\psi'_m  - {1\over 2} \f \psi''_m \Big]\pa_x u dxdt.
\end{equation}
It is clear that $\lim_{m\to\infty} \sE^m_{0,t}[u(t, X_t)] = \sE_{0,t}[u(t,X_t)]$. Note that, by \eqref{mugrowth} and \eqref{paxubound},  we have
\beaa
|\mu|\le C[1+|x|],\quad  \pa_p\f \rho \le   {C\rho \over G[1-G]} \le C[1+|x|],\quad |\pa_x u|\le Ce^{-x^2}.
\eeaa
Then, by sending $m\to \infty$ in \eqref{dsEt2} and  applying the dominated convergence theorem, we obtain \eqref{dsE} and hence the theorem.  
\end{proof}

\medskip

\begin{remark} \label{rem-singular}
In the definition of $\Phi$ (see \eqref{Phi-cont}) we require that the initial time $s$ is strictly positive. In fact, when $s=0$ the distribution of $X_s$ becomes degenerate, and thus $\mu$ may have singularities.  For example, assume $\f(t,\cdot)= \f(\cdot)$ is independent of $t$ and $b\equiv0$, $x_0=0$. Then
\beaa
\mu(t, x) = -{\f''(G(t,x))\over 2 \f'(G(t,x))} {1\over 2\sqrt{\pi t}} e^{-{x^2\over 2t}}.
\eeaa
It is not even clear if the following SDE is wellposed in general: 
\beaa
\tilde X_t =  \int_0^t \mu(s, \tilde X_s) ds + B_t.
\eeaa
Correspondingly, if we consider the following PDE on $(0, T]\times \hR$:
\beaa
\sL u(t,x) = 0, ~ (t,x) \in (0, T]\times \hR,\quad u(T,x) = g(x).
\eeaa
then it is not clear whether or not $\lim_{(t,x) \to (0, 0)} u(t,x)$ exists.
\end{remark}

 Unlike Theorem \ref{thm-ddf} in the discrete case, surprisingly here the time-consistent dynamic distortion function is {\it not} unique. Let $\check{\Phi}$ be an arbitrary time-consistent dynamic distortion function for $0<s<t\le T$ (not necessarily consistent with $\f_t$ when $s=0$ at this point). Fix $0<t \le T$. For any $s\in (0, t]$ and  $g\in \cI$, define 
\bea
\label{hatu}
\check{u}(s, x) :=  \int_\hR \check{\Phi}(s,t,x; \hP( g(X_t) \ge y|X_s=x)) ds.
\eea
The corresponding $\{\check{\sE}_{s, t}\}$ is time-consistent, i.e., the tower property holds. Suppose $\check{\Phi}$ defines via \eqref{PhiQ} a $\mathbb{Q}$-diffusion $\check{X}$ with coefficients $\check{\mu}$, $\check{\sigma}$, i.e., 
\bea \label{checkPhi}
\left.\ba{c}
\check\Phi(s,t,x; \hP(X_t \ge y|X_s=x)) = \hP(\check X^{s,x}_t\ge y),\\
\mbox{where}\quad \check X^{s,x}_t = x + \int_s^t \check \mu(r, \check X^{s,x}_r) dB_r+\int_s^t \check \sigma(r, \check X^{s,x}_r) dB_r,\quad \hP\mbox{-a.s.}
\ea\right.
\eea
and $\check u$ satisfies the following PDE corresponding to the infinitesimal generator of $\check{X}$:
\bea
\label{hatPDE}
\pa_t \check u + {1\over 2} \check \sigma^2 \pa_{xx}\check u + \check \mu \pa_x \check u =0,~ (s, x) \in (0, t]\times \hR;\quad \check u(t,x) = g(x).
\eea
We have the following more general result.

\begin{theorem}
\label{thm-limit2}
Let Assumptions \ref{assum-b} and \ref{assum-Gphi} hold, and $\check \Phi$ be an arbitrary smooth time-consistent (for $t>0$) dynamic distortion function corresponding to \eqref{checkPhi}. Suppose $\check{\mu}$ and $\check{\sigma}$ are sufficiently smooth such that $\check{u}$ is smooth and integration by parts in \eqref{dsEthat} below goes through. Then $\check \Phi$ is consistent with the initial condition $\check\Phi(0,t, x; p) = \f_t(p)$ if and only if 
	\bea
	\label{hatmu}
	\check \mu = b + \check \sigma \pa_x \check \sigma  + {1\over 2}[\check\sigma^2 -1]{\pa_x\rho \over \rho}+ {\pa_t \f \over \pa_p\f \rho} - {\check\sigma^2 \pa_{pp}\f \rho \over 2\pa_p\f}.
	\eea
	
	In particular, if we restrict to the case $\check\sigma = 1$, then $\check \mu = \mu$ and hence $\check\Phi = \Phi$ is unique.	
\end{theorem}
\begin{proof}  
The consistency of $\check \Phi$ with the initial condition $\Phi(0,t, x; p) = \f_t(p)$ is equivalent to  \eqref{dsE} for $\check u$, where $\check u$ is the solution to PDE \eqref{hatPDE} on $(0, T]$ with terminal condition $g$. Similar to \eqref{dsEt}, we have

 \bea
\label{dsEthat}
&&{d\over dt} \sE_{0,t}\big[\check u(t, X_t)]\big] = \int_\hR \Big[[\pa_t \f + \pa_p \f \pa_t G] \pa_x \check u + \f \pa_{tx} \check u\Big] dx \nonumber\\
&&=  \int_\hR \Big[[\pa_t \f + \pa_p \f \pa_t G] \pa_x \check u + \pa_p \f  \rho \pa_{t} \check u\Big] dx\nonumber\\
&&= \int_\hR \Big[\big[\pa_t \f + \pa_p \f [b \rho - {1\over 2}\pa_x \rho]\big] \pa_x \check u  - \pa_p \f  \rho \big[  {1\over 2} \check \si^2 \pa_{xx}\check u + \check \mu \pa_x u\big] \Big] dx\nonumber\\
&&= \int_\hR \Big[\big[\pa_t \f + \pa_p \f [b \rho - {1\over 2}\pa_x \rho]- \pa_p \f  \rho \check \mu\big] \pa_x \check u   \\
&&\quad +{1\over 2} \pa_{x} \check u \big[- \pa_{pp} \f \rho^2 \check \si^2 + \pa_p \f \pa_x \rho \check \si^2 + 2 \pa_p\f \rho \check\si \pa_x \check \si  \big] \Big] dx\nonumber\\
&&=\int_\hR \Big[b + \check \sigma \pa_x \check \sigma + {\pa_t \f \over \pa_p\f \rho} - {\check\si^2 \pa_{pp}\f \rho \over 2\pa_p\f} + {1\over 2}[\check\si^2 -1]{\pa_x\rho \over \rho}\big] -\check \mu\Big]\pa_p\f \rho \pa_x\check u   dx.\nonumber
\eea
Since $\pa_p\f \rho >0$, and $g$  and hence $\check u$ is arbitrary, we get the equivalence of \eqref{dsE} and \eqref{hatmu}.  
\end{proof}

\medskip

\begin{remark} \label{rem-uniqueness}  {\ }
\begin{enumerate}
	\item[(i)] When $\check{\sigma} \not\equiv 1$, the law of $\check X^{s,x}_t$ can be singular to the conditional law of $X_t$ given $X_s = x$. That is, the agent may distort the probability so dramatically that the distorted probability is singular to the original one. For example, some event which is null under the original probability may be distorted into a positive or even full measure, so the agent could be worrying too much on something which could never happen, which does not seem to be reasonable in practice. Our result says, if we exclude this type of extreme distortions, then for given $\{\f_t\}$, the time-consistent dynamic distortion function $\Phi$ is unique.  
	\item[(ii)] In the discrete case in Section \ref{sect:GBT}, due to the special structure of binomial tree, we always have $|X_{t_{i+1}}-X_{t_i}|^2 = h$. Then for any possible $\hQ$, we always have $\hE^{\hQ}\big[|X_{t_{i+1}}-X_{t_i}|^2 \big| X_{t_i} = x_{i,j}\big]= h$. This, in the continuous time model, means $\check \si \equiv 1$. This is why we can obtain the uniqueness in Theorem \ref{thm-ddf}.
\end{enumerate}
\end{remark}

\subsection{Rigorous proof of the convergence.}
We note that Theorem \ref{thm-limit} already gives the definition of the desired time-consistent conditional expectation for the constant diffusion case. Nevertheless, it is still worth asking whether  the discrete system in Section \ref{sec:approx} indeed converges to the continuous time system in Section \ref{sect:limit}, especially from the perspective of numerical approximations. We therefore believe that a detailed convergence analysis, which we now describe, is interesting in its own right. 

For each $N$, denote $h := h_N:= {T\over N}$, and $t_i := t^N_i :=  ih$, $i=0,\ldots, N$, as in Section \ref{sec:approx}.  Consider the notations in \eqref{BTN} and \eqref{GQPhiN}, and denote
\bea
\label{rhoN}
\rho^N_{i,j} :=   \hP^N(X^N_{t_i} = x_{i,j})\slash (2\sqrt{h}).
\eea

\begin{proposition}
	\label{prop-Grho}
	Under Assumptions \ref{assum-b}, for any sequence $(t^N_i, x^N_{i,j}) \to (t,x) \in (0, T]\times \hR$, we have $G^N_{i,j} \to G(t,x)$ and $\rho^N_{i,j} \to \rho(t,x)$ as  $N\to \infty$.
\end{proposition}

Again we postpone this proof to Section \ref{sect-density}.

\begin{theorem}
	\label{thm-conv}
	Let Assumptions \ref{assum-b} and \ref{assum-Gphi} hold, and $g\in \cI$.  For each $N$, consider the notations in \eqref{BTN} and \eqref{GQPhiN}, and define by backward induction as in \eqref{ui}:
	\begin{equation}
	\label{uNi}
	u^N_N(x) := g(x),~ u^N_i(x_{i,j}) := q^{N, +}_{i,j} u^N_{i+1}(x_{i+1, j+1}) + q^{N, -}_{i,j} u^N_{i+1}(x_{i+1, j}),~  i=N-1, \ldots, 0.
	\end{equation}
	Then, for any $(t,x) \in (0, T]\times \hR$ and any sequence $(t^N_i, x^N_{i,j}) \to (t,x)$, we have 
	\bea
	\label{uNconv}
	\lim_{N\to\infty} u^N_i(x_{i,j}) = u(t,x).
	\eea
\end{theorem}
\begin{proof}  
Define
\beaa
\overline{u}(t,x) := \limsup_{N\to \infty, t_i\downarrow t, x_{i,j}\to x} u^N_i(x_{i,j}),\quad \underline{u}(t,x) := \liminf_{N\to \infty, t_i\downarrow t, x_{i,j}\to x} u^N_i(x_{i,j}).
\eeaa
We shall show that $\overline{u}$ is a viscosity subsolution and $\underline{u}$ a viscosity supersolution of PDE \eqref{PDE}. By the comparison principle of the PDE \eqref{PDE} we have $\overline{u} = \underline{u} = u$, which implies \eqref{uNconv} immediately.

We shall only prove  $\overline{u}$ is a viscosity subsolution. The viscosity supersolution property of $\underline{u}$ can be proved similarly. Fix $(\overline{t}, \overline{x}) \in (0, T] \times \hR$. Let $w$ be a smooth test function at $(\overline{t}, \overline{x})$ such that $[w- \overline{u}](\overline{t}, \overline{x}) = 0 \le [w-u](t,x)$ for all $(t,x) \in [\overline{t}, T] \times \hR$ satisfying $t- \ol t \le \delta^2, |x-\ol x|\le \delta$  for some $\delta>0$. Introduce 
\bea
\label{tildew}
\tilde w(t,x) := w(t,x) + \delta^{-5} [|t-\ol t|^2 + |x-\ol x|^4].
\eea
Then
\beaa
[\tilde w - \ol u](\ol t, \ol x) = 0 < {1\over  C\delta} \le \inf_{ {\delta^2 \over 2} \le |t-\ol t| + |x-\ol x|^2 \le \delta^2} [\tilde w - \ol u](t,x).
\eeaa
By the definition of $\ol u(\ol t, \ol x)$,  by otherwise choosing a subsequence of $N$, without loss of generality we assume there exist $(i_N, j_N)$ such that $t_{i_N} \downarrow \ol t$, $x_{i_N, j_N} \to \ol x$, and $\lim_{N\to\infty} u^N_{i_N}(x_{i_N, j_N}) = \ol u(\ol t, \ol x)$. Since $\ol u$ and $u^N$ are bounded, for $\delta$ small, we have
\beaa
c_N:=[\tilde w - u^N](t_{i_N}, x_{i_N, j_N})  < {1\over  2C\delta} \le  \inf_{ {\delta^2 \over 2}\le |t_i-\ol t| + |x_{i,j}-\ol x|^2 \le \delta^2} [\tilde w -  u^N](t_i,x_{i,j}).
\eeaa
Denote
\beaa
c^*_N := \inf_{t_{i_N} \le t_i\le \ol t + {\delta^2\over 2},   |x_{i,j}-\ol x|^2 \le \delta^2} [\tilde w -  u^N](t_i,x_{i,j}) = [\tilde w -  u^N](t_{i^*_N},x_{i^*_N,j^*_N})  \le c_N.
\eeaa
Then clearly $|t_{i^*_N} - \ol t| +| x_{i^*_N,j^*_N} - \ol x| < {\delta^2\over 2}$. Moreover, by a compactness argument, by otherwise choosing a subsequence, we may assume $(t_{i^*_N}, x_{i^*_N,j^*_N}) \to (t_*, x_*)$.  Then
\beaa
0 &=& \lim_{N\to \infty} c_N \ge \limsup_{N\to \infty} [\tilde w -  u^N](t_{i^*_N},x_{i^*_N,j^*_N}) = \tilde w(t_*, x_*) - \liminf_{N\to \infty} u^N(t_{i^*_N},x_{i^*_N,j^*_N})\\
&\ge& \tilde w(t_*, x_*)  - \ol u(t_*, x_*) \ge \delta^{-5} [|t_*-\ol t|^2 + |x_*-\ol x|^4].
\eeaa
That is, $(t_*, x_*) = (\ol t, \ol x)$, namely 
\bea
\label{xNconv}
\lim_{N\to \infty} (t_{i^*_N}, x_{i^*_N,j^*_N}) = (\ol t, \ol x).
\eea
Note that
\beaa
\tilde w(t_{i^*_N},x_{i^*_N,j^*_N}) &=& u^N(t_{i^*_N},x_{i^*_N,j^*_N}) + c^*_N\\
&=& q^{N, +}_{i^*_N,j^*_N} u^N(t_{i^*_N+1},x_{i^*_N+1,j^*_N+1}) + q^{N, -}_{i^*_N,j^*_N} u^N(t_{i^*_N+1},x_{i^*_N+1,j^*_N})+ c^*_N\\
&\le&  q^{N, +}_{i^*_N,j^*_N} \tilde w(t_{i^*_N+1},x_{i^*_N+1,j^*_N+1}) + q^{N, -}_{i^*_N,j^*_N} \tilde w(t_{i^*_N+1},x_{i^*_N+1,j^*_N}).
\eeaa
Then, denoting $(i, j) := (i^*_N, j^*_N)$ for notational simplicity, we have
\bea
\label{tildewest}
0&\le&q^{N, +}_{i,j}\big[\tilde w(t_{i+1},x_{i+1,j+1}) - \tilde w(t_{i},x_{i,j})\big]+ q^{N, -}_{i,j} \big[\tilde w(t_{i+1},x_{i+1,j})-\tilde w(t_{i},x_{i,j})\big]\nonumber\\
&=& q^{N, +}_{i,j}[\pa_t \tilde w(t_i, x_{i,j}) h + \pa_x \tilde w(t_i, x_{i,j}) \sqrt{h} + {1\over 2}  \pa_{xx} \tilde w(t_i, x_{i,j}) h] \\
&&  + q^{N, -}_{i,j}[\pa_t \tilde w(t_i, x_{i,j}) h - \pa_x \tilde w(t_i, x_{i,j}) \sqrt{h} + {1\over 2}  \pa_{xx} \tilde w(t_i, x_{i,j}) h]  + o(h)\nonumber\\
&=&  [\pa_t \tilde w(t_i, x_{i,j})  +  {1\over 2}  \pa_{xx} \tilde w(t_i, x_{i,j}) ] h + [q^{N, +}_{i,j}  - q^{N, -}_{i,j}]\pa_x \tilde w(t_i, x_{i,j}) \sqrt{h} +  o(h).\nonumber
\eea
Note that
\beaa
&& q^{N, +}_{i,j}  - q^{N, -}_{i,j} = 1+2 \frac{\varphi_{t_{i+1}}(G^N_{i+1, j+1} ) - \varphi_{t_i}(G^N_{i, j })}{\varphi_{t_i}(G^N_{i, j}) - \varphi_{t_i}(G^N_{i, j + 1})};\\
&&\varphi_{t_i}(G^N_{i, j}) - \varphi_{t_i}(G^N_{i, j + 1}) =  \varphi_{t_i}(G^N_{i, j}) - \varphi_{t_i}(G^N_{i, j } - 2 \rho^N_{i,j} \sqrt{h}) = \pa_p \f(t_i, G^N_{i,j})2 \rho^N_{i,j} \sqrt{h}  +  o(\sqrt{h});\\
&&\varphi_{t_{i+1}}(G^N_{i+1, j+1} ) - \varphi_{t_i}(G^N_{i, j }) = \varphi_{t_{i+1}}\big(G^N_{i, j} -  2 \rho^N_{i,j} \sqrt{h}  p_{i,j}^- \big) - \varphi_{t_i}(G^N_{i, j}) \\
&&\quad = \pa_t \f_{t_i}(G^N_{i,j}) h - \pa_p\f_{t_i}(G^N_{i,j}) 2 \rho^N_{i,j} \sqrt{h}  p_{i,j}^- + {1\over 2} \pa_{pp} \f_{t_i}(G^N_{i,j}) [2 \rho^N_{i,j} \sqrt{h}  p_{i,j}^-]^2  + o(h)\\
&&\quad = \pa_t \f_{t_i}(G^N_{i,j}) h - \pa_p\f_{t_i}(G^N_{i,j}) \rho^N_{i,j} \sqrt{h}  [1-b_{i,j}\sqrt{h}] + {1\over 2} \pa_{pp} \f_{t_i}(G^N_{i,j}) [\rho^N_{i,j}]^2h     + o(h).
\eeaa 
Then, denoting $G_{i,j} := G(t_i, x_{i,j})$, $\rho_{i,j} := \rho(t_i, x_{i,j})$ and by Proposition \ref{prop-Grho},
\beaa
q^{N, +}_{i,j}  - q^{N, -}_{i,j}
&=& { \pa_t \f_{t_i}(G^N_{i,j}) h  + \pa_p\f_{t_i}(G^N_{i,j})  \rho^N_{i,j}  b_{i,j} h   - {1\over 2}  \pa_{pp} \f_{t_i}(G^N_{i,j}) [ \rho^N_{i,j}  ]^2h + o(h) \over  \pa_p \f(t_i, G^N_{i,j})  \rho^N_{i,j} \sqrt{h}  + o(\sqrt{h})}\\
&=& \Big[b_{i,j} + { \pa_t \f_{t_i}(G^N_{i,j})      -{1\over 2}  \pa_{pp} \f_{t_i}(G^N_{i,j}) [ \rho^N_{i,j}  ]^2  \over  \pa_p \f(t_i, G^N_{i,j})  \rho^N_{i,j}   } + o(1)\Big]\sqrt{h}\\
&=& \Big[b_{i,j} + { \pa_t \f_{t_i}(G_{i,j})      -{1\over 2}  \pa_{pp} \f_{t_i}(G_{i,j}) [ \rho_{i,j}  ]^2  \over  \pa_p \f(t_i, G_{i,j})  \rho_{i,j}   } + o(1)\Big]\sqrt{h}\\
&=& \big[\mu(t_i, x_{i,j})  + o(1)\big]\sqrt{h}.
\eeaa
Thus, by \eqref{tildewest} and \eqref{xNconv},
\beaa
0&\le& \Big[\pa_t \tilde w(t_i, x_{i,j})  +  {1\over 2}  \pa_{xx} \tilde w(t_i, x_{i,j}) +\mu(t_i, x_{i,j}) \pa_x \tilde w(t_i, x_{i,j})\Big] h +  o(h)\\
&=&\Big[\pa_t \tilde w(\ol t, \ol x)  +  {1\over 2}  \pa_{xx} \tilde w(\ol t, \ol x) +\mu(\ol t, \ol x) \pa_x \tilde w(\ol t, \ol x)\Big] h +  o(h).
\eeaa
This implies $\sL \tilde w (\ol t, \ol x) \ge 0$. By \eqref{tildew}, it is clear that $\sL w (\ol t, \ol x) = \sL \tilde w (\ol t, \ol x)$. Then $\sL w (\ol t, \ol x) \ge 0$, thus $\ol u$ is a viscosity subsolution at $(\ol t, \ol x)$  
\end{proof}

\section{The general diffusion case.}  \label{sect:generaldiffusion}

In this section we consider a general diffusion process given by the SDE:
\bea
\label{Xdiffusion}
X_t = x_0 + \int_0^t b(s, X_s) ds + \int_0^t \sigma(s, X_s) dB_s,\quad \hP\mbox{-a.s.}
\eea
Provided that $\si$ is non-degenerate, this problem can be transformed back to \eqref{X2}: 
\begin{equation}
\label{Psi}
\hat X_t := \psi(t, X_t),  \quad \hat x_0 := \psi(0, x_0),\quad \mbox{where} \quad \psi(t,x) := \int_0^x {dy\over \si(t, y)}.
\end{equation}
Then, by a simple application of  It\^{o}'s formula, we have
\begin{equation}
\label{hatX}
\hat X_t = \hat x_0 + \int_0^t \hat b(s, \hat X_s) ds + B_t,\quad \mbox{where}\quad \hat b (t, x):= \big[\pa_t \psi + {b\over \si} - \frac{1}{2} \pa_x \sigma\big](t, \psi^{-1}(t,x)).
\end{equation}
Here $\psi^{-1}$ is the inverse mapping of $x\mapsto \psi(t, x)$. Denote
\begin{equation}
\label{hatGrho}
G(t,x) := \hP(X_t \ge x), \quad \rho := -\pa_x  G,\quad \hat G(t,x) := \hP(\hat X_t \ge x), \quad \hat \rho := -\pa_x \hat G.
\end{equation}

To formulate a rigorous statement we shall make the following assumption.

\begin{assumption}
	\label{assum-si}
	The functions $b,\si$ are sufficiently smooth and both $b, \si$ and the required derivatives are bounded. Moreover, $\si \ge c_0 >0$. 
\end{assumption}

The following result is immediate and we omit the proof.
\begin{lemma}
	\label{lem-hatGrho}

	Under Assumption \ref{assum-si}, we have
	
	\begin{enumerate}
		\item[(i)] the $\hat b$ defined in \eqref{hatX} satisfies Assumption \ref{assum-b};
	    \item[(ii)] $G(t,x) = \hat G(t, \psi(t, x)), \rho(t,x) = {\hat \rho(t,\psi(t,x))\over \si(t,x)}$ are sufficiently smooth and satisfy \eqref{Gbound}.
	\end{enumerate}
\end{lemma}

Here is the main result of this section.

\begin{theorem}
	\label{thm-si}
	Assume Assumptions \ref{assum-si} and \ref{assum-Gphi} hold. Let $\check \Phi$ be a time-consistent dynamic distortion function determined by \eqref{checkPhi} for $0<s<t\le T$, where $\check \si$ and $\check \mu$ satisfy the same technical requirements as in Theorem \ref{thm-limit2}. Then $\check\Phi$ is consistent with initial condition $\check\Phi(0,t,x_0; p) = \f_t(p)$ if and only if 
	\bea
	\label{checkmusi}
	\check \mu =b -  \si \pa_x \sigma +\check \si \pa_x \check\si + {1\over 2}[\check\si^2 -\si^2]{\pa_x \rho \over \rho }+{\pa_t \f (t, G(t,x))\over \pa_p\f (t, G(t,x))\rho} - {\check\si^2\rho \pa_{pp}\f (t, G(t,x))  \over 2\pa_p\f(t, G(t,x))}.
	\eea
	
	In particular, if we require $\check \si = \si$, then $\check\Phi$ is unique with 
	\begin{equation}
	\label{checkmu}
	\check \mu(t,x) = \mu(t,x) := b(t,x) +   {\pa_t \f (t, G(t,x))  - {1\over 2} \pa_{pp}\f(t, G(t,x)) \rho^2\si^2(t,x) \over \pa_p \f(t, G(t,x)) \rho(t,x)}.
    \end{equation}
\end{theorem}
\begin{proof}  
Let $g\in \cI$ and $\check u$ be the solution to PDE \eqref{hatPDE} on $(0, T]\times \hR$ with terminal condition $g$. Then $\check\Phi$ is consistent with initial condition $\check\Phi(0,t,x_0; p) = \f_t(p)$ means the mapping $t\in (0, T]\mapsto \sE_{0, t}[\check u(t, X_t)]$ is a constant. Note that
\beaa
\sE_{0, t}[\check u(t, X_t)] = \int_\hR \f_t( \hP(X_t \ge x)) \pa_x \check u(t,x ) dx =  \int_\hR \f_t( \hP(\hat X_t \ge \psi(t,x))) \pa_x \check u(t,x ) dx. 
\eeaa
Denote $\hat x:= \psi(t, x)$. Then
\begin{equation}
\label{checkuhat}
\sE_{0, t}[\check u(t, X_t)] =  \int_\hR \f_t( \hP(\hat X_t \ge \hat x)) \pa_{\hat x} \hat u(t, \hat x) d\hat x,\quad\mbox{where}\quad \hat u(t, \hat x) :=  \check u(t, \psi^{-1}(t,\hat x) ).
\end{equation}

Note that $\check u(t,x) = \hat u(t, \psi(t, x))$. Then
\beaa
\pa_t \check u = \pa_t \hat u + \pa_{\hat x}\hat u \pa_x \psi,\quad \pa_x \check u = \pa_{\hat x} \hat u \pa_x \psi,\quad \pa_{xx} \check u = \pa_{\hat x\hat x} \hat u (\pa_x \psi)^2 + \pa_{\hat x}\hat u \pa_{xx}\psi,
\eeaa
and thus PDE \eqref{hatPDE} implies
\beaa
0&=& \big[\pa_t \hat u + \pa_{\hat x}\hat u \pa_t \psi\big] + {1\over 2} \check \si^2  \big[\pa_{\hat x\hat x} \hat u (\pa_x \psi)^2 + \pa_{\hat x}\hat u \pa_{xx}\psi\big] + \check \mu  \pa_{\hat x} \hat u \pa_x \psi\\
&=& \pa_t \hat u +  {1\over 2} (\check \si\pa_x \psi)^2 \pa_{\hat x\hat x} \hat u  + \big[\pa_t \psi + {1\over 2} \check \si^2 \pa_{xx}\psi +  \check \mu\pa_x \psi\big] \pa_{\hat x} \hat u.
\eeaa
Recall \eqref{hatX} and \eqref{checkuhat} and note that $G(t,x) = \hat G(t, \psi(t, x))$. Applying Theorem \ref{thm-limit2} we see that the required time-consistency is equivalent to
\begin{eqnarray} \label{psi1}
&\pa_t \psi + {1\over 2} \check \si^2 \pa_{xx}\psi +  \check \mu\pa_x \psi \\
&= \hat b + (\check \si \pa_x \psi)\pa_{\hat x} (\check \si \pa_x \psi) + {1\over 2}[(\check\si\pa_x\psi)^2 -1]{\pa_{\hat x}\hat \rho \over \hat \rho}+ {\pa_t \f (t, G(t,x))\over \pa_p\f (t, G(t,x)) \hat \rho} - {(\check\si\pa_x\psi)^2 \pa_{pp}\f (t, G(t,x)) \hat \rho \over 2\pa_p\f(t, G(t,x))}.\nonumber
\end{eqnarray}
Note that
\beaa
&\pa_x \psi \pa_{\hat x} (\check \si \pa_x \psi) = \pa_{\hat x} (\check \si \pa_x \psi),\quad \pa_x \psi = {1\over \si},\quad \pa_{xx} \psi = -{\pa_x \si\over \si^2},\\
& \hat \rho(t, \psi(t,x)) = \rho\si(t,x),\quad \pa_{\hat x}\hat\rho = [\pa_x \rho \si + \rho \pa_x \si] \si.
\eeaa
Then \eqref{psi1} is equivalent to
\beaa
&&\pa_t \psi - {\check \si^2 \pa_x \si\over 2\si^2}+  {\check \mu\over \si} =\big[\pa_t \psi + {b\over \si} - \frac{1}{2} \pa_x \sigma\big] +\big[{\check \si \pa_x \check\si\over \si} -  {\check \si^2 \pa_x \si\over \si^2}\big]\\
&&\quad+ {1\over 2}[({\check\si\over \si})^2 -1][{\pa_x \rho\si \over \rho} + \pa_x\si] + {\pa_t \f (t, \psi(t,x))\over \pa_p\f (t, G(t,x))\rho\si} - {\check\si^2\rho \pa_{pp}\f (t, G(t,x))  \over 2\pa_p\f(t, G(t,x))\si}.
\eeaa
This implies \eqref{checkmusi} immediately.  
\end{proof}

\medskip

\begin{remark} \label{rem-general}
In this remark we investigate possible discretization for  the general SDE \eqref{Xdiffusion},  in the spirit of Section \ref{sec:approx}. Note that
\beaa
X_{t_{i+1}} \approx X_{t_i} + b(t_i, X_{t_i}) h +  \si(t_i, X_{t_i}) [B_{t_{i+1}} - B_{t_i}].
\eeaa
For a desired approximation $X^N$, we would expect
\begin{equation}
\label{trinomial}
\hE\big[ X^N_{t_{i+1}} - X^N_{t_i} \big| X^N_{t_i} = x \big] = b(t_i, x) h + o(h),~ \hE\big[ (X^N_{t_{i+1}} - X^N_{t_i})^2 \big| X^N_{t_i} = x \big] = \si^2(t_i, x) h + o(h).
\end{equation}
However, for the binomial tree in Figure \ref{fig:general.tree1}, at each node $x_{i,j}$ there is only one parameter $p_{i,j}^+$ and in general we are not able to match both the drift and the volatility. To overcome this, we have three natural choices:

\begin{enumerate}
	\item[(i)] The first one is to use trinomial tree approximation:  assuming $0<\si\le C_0$, we have
	\beaa
	&x_{i,j} = C_0 j \sqrt{h},  ~ j=-i,\ldots, i,\quad \hP\big(X^N_{t_{i+1}} = x_{i+1, j+1} \big| X^N_{t_i} = x_{i,j}\big) = p_{i,j}^+,&\\
	& \hP\big(X^N_{t_{i+1}} = x_{i+1, j-1} \big| X^N_{t_i} = x_{i,j}\big) = p_{i,j}^-,\quad \hP\big(X^N_{t_{i+1}} = x_{i+1, j} \big| X^N_{t_i} = x_{i,j}\big) = p_{i,j}^0 := 1-p_{i,j}^+ - p_{i,j}^-.&
	\eeaa
	See the left figure in Figure \ref{fig:trinomial} for the case $N=2$. Then, by choosing appropriate $p_{i,j}^+, p_{i,j}^-$, one may achieve \eqref{trinomial}. However, note that the trinomial tree has crossing edges, and they may destroy the crucial monotonicity property we used in the previous section, as we saw in Remark \ref{rem-nonmonotone} (iii) and Example \ref{eg-crossing}.
	\begin{figure}[h]
		\centering
		\begin{tikzpicture}[>=stealth]
		\matrix (tree) [
		matrix of nodes,
		minimum size=0.5cm,
		column sep=1.7cm,
		row sep=0.5cm,
		]
		{
			&                 &$x_{2,2}$ \\
			&$x_{1,1}$&$x_{2,1}$ \\
			$x_{0,0}$&$x_{1,0}$&$x_{2,0}$ \\
			&$x_{1,-1}$&$x_{2,-1}$ \\
			&                 &$x_{2,-2}$ \\
		};
		\draw[->] (tree-3-1) -- (tree-2-2) node [midway,above] {};
		\draw[->] (tree-3-1) -- (tree-3-2) node [midway,above] {};
		\draw[->] (tree-3-1) -- (tree-4-2) node [midway,above] {};
		\draw[->] (tree-2-2) -- (tree-1-3) node [midway,above] {};
		\draw[->] (tree-2-2) -- (tree-2-3) node [midway,below] {};
		\draw[->] (tree-2-2) -- (tree-3-3) node [midway,below] {};
		\draw[->] (tree-3-2) -- (tree-2-3) node [midway,above] {};
		\draw[->] (tree-3-2) -- (tree-3-3) node [midway,above] {};
		\draw[->] (tree-3-2) -- (tree-4-3) node [midway,above] {};
		\draw[->] (tree-4-2) -- (tree-3-3) node [midway,below] {};
		\draw[->] (tree-4-2) -- (tree-4-3) node [midway,below] {};
		\draw[->] (tree-4-2) -- (tree-5-3) node [midway,below] {};
		\end{tikzpicture} \quad
		\begin{tikzpicture}[>=stealth]
		\draw  (0, 0) -- (2, 1);
		\draw  (0, 0) -- (2, -1);
		\draw  (2, 1) -- (4, 2.7);
		\draw  (2, 1) -- (4, -0.7); 
		\draw  (2, -1) -- (4, 0);
		\draw  (2, -1) -- (4, -2);
		\node [right] at (4, 2.7) {$x_2 + \sigma_2 \sqrt{h}$};
		\node [right] at (4, -0.7) {$x_2 - \sigma_2 \sqrt{h}$};
		\node [right] at (4, -2) {$x_1 - \sigma_1 \sqrt{h}$};
		\node [right] at (4, 0) {$x_1 + \sigma_1 \sqrt{h}$};
		\node [below] at (2, 1) {$x_2$};
		\node [below] at (2, -1) {$x_1$};
		\node [below] at (0, 0) {$x_0$};
		
		\end{tikzpicture}
		\caption{Left: trinomial tree; Right: binary tree} 
		\label{fig:trinomial}
	\end{figure}
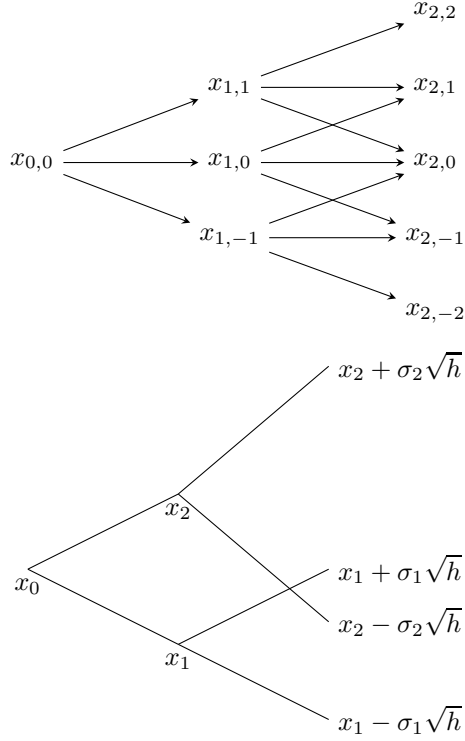
	\item[(ii)] The second choice  is to use the binary tree approximation, see the right figure in Figure \ref{fig:trinomial} for the case $N=2$, where
	$x_1 = x_0 - \si(t_0, x_0) \sqrt{h}$, $x_2 = x_0+\si(t_0, x_0) \sqrt{h}$, $\si_1 = \si(t_1, x_1)$, $\si_2 =\si(t_1, x_2)$.
	But again there are crossing edges and thus the monotonicity property is violated. 
	\item[(iii)] The third choice, which indeed works well, is to utilize the transformation \eqref{Psi}. Let $\hat X^N$ be the discretization for $\hat X$ in \eqref{hatX}, as introduced in Section \ref{sec:approx}. Then $X^N := \psi_{t_i}^{-1}(\hat X^N_{t_i})$ will serve for our purpose. We skip the details here.
\end{enumerate}
\end{remark}

\section{Analysis of the density.} 
\label{sect-density}

In this section we prove Propositions \ref{prop-tail} and  \ref{prop-Grho}. The estimates rely on the following representation formula for $\rho$ by using the Brownian bridge. The result is a direct consequence of Karatzas-Shreve \cite[Section 5.6. Exercise 6.17]{KS}, and holds true in multidimensional case as well.

\begin{proposition}
	\label{prop-rep}
	Assume $b$ is  bounded. Then we have the following representation formula:
	\bea
	\label{rep}
	\left.\ba{c}
	\rho(t,x) = {1\over \sqrt{2\pi t}} \exp\Big(-{(x-x_0)^2\over 2t}  + I(t,x)\Big),\quad t>0,\quad \mbox{where}\\
	 \bar M^t_s := \int_0^s {dB_r \over t-r},\quad \bar X^{t,x}_s := x_0 + [x-x_0]{s\over t} +[t-s] \bar M^t_s, \quad 0\le s< t;\\
	e^{I(t,x)} := \hE\Big[e^{ \int_0^t b(s, \bar X^{t,x}_s) dB_s  + \int_0^t [ (x-x_0) b(s, \bar X^{t,x}_s)  - b(s, \bar X^{t,x}_s) \bar M^t_s- {1\over 2}  |b(s, \bar X^{t,x}_s)|^2]ds }\Big].
	\ea\right.
	\eea
\end{proposition}

\begin{proof} Since we will use the arguments, in particular that for \eqref{Mintegrable} below, in the proof of Proposition \ref{prop-Grho}, we provide a detailed proof here. For notational simplicity, let's assume $t=1$ and $x_0=0$. Then \eqref{rep} becomes:
\bea
\label{rep1}
\left.\ba{c}
\rho(1,x) = {1\over \sqrt{2\pi }} \exp\Big(-{x^2\over 2}  + I(x)\Big),\quad \mbox{where}\\
\bar M_s := \int_0^s {dB_r \over 1-r},\quad \bar X^{x}_s :=  x s +[1-s] \bar M_s, 0\le s< 1;\\
e^{I(x)} := \hE\Big[e^{ \int_0^1 b(s, \bar X^{x}_s) dB_s  + \int_0^1 [ x b(s, \bar X^{x}_s)  - b(s, \bar X^{x}_s) \bar M_s- {1\over 2}  |b(s, \bar X^{x}_s)|^2]ds }\Big].
\ea\right.
\eea

We first show that the right side of the last line in \eqref{rep1} is integrable. Since $b$ is bounded, it suffices to prove the following (stronger) claim: for any $C>0$ and $\a\in (0, 2)$,
\bea
\label{Mintegrable}
\hE\Big[e^{C\int_0^1|\bar M_s|^\a ds}\Big] <\infty.
\eea
Indeed, by time change $s={t\over 1+t}$, we have $\int_0^1|\bar M_s|^\a ds = \int_0^\infty {|\bar M_{t/(1+t)}|^\a\over (1+t)^2 } dt$. Since
\beaa
\hE\Big[|\bar M_{t\over 1+t}|^2\Big] = \int_0^{t\over 1+t} {dr\over (1-r)^2} = t.
\eeaa
by Levy's characterization we see that $t\mapsto \bar M_{t/(1+t)}$ is a Brownian motion.  Then
\beaa
\hE\Big[e^{C\int_0^1|\bar M_s|^\a ds}\Big]  = \hE\Big[e^{C\int_0^\infty {|B_t|^\a \over (1+t)^2 } dt}\Big] = \sum_{n=0}^\infty{C^n\over n!} \hE\Big[\Big(\int_0^\infty {|B_t|^\a\over (1+t)^2 } dt\Big)^n\Big].
\eeaa
Note that
\beaa
\int_0^\infty {|B_t|^\a\over (1+t)^2 } dt \le \sup_{t\ge 0} {|B_t|^\a\over (1+t)^{2+\a \over 4}} \int_0^\infty {dt\over (1+t)^{1+{2-\a\over 4}} }  ={4\over 2-\a} \sup_{t\ge 0} {|B_t|^\a\over (1+t)^{2+\a \over 4}}.
\eeaa
Then, for a generic constant $C$,
\beaa
&&\hE\Big[e^{C\int_0^1|\bar M_s|^\a ds}\Big]  \le \sum_{n=0}^\infty{C^n\over n!} \hE\Big[\sup_{t\ge 0} {|B_t|^{n\a}\over (1+t)^{n(2+\a)\over 4}}\Big]\\
&&\le\sum_{n=0}^\infty{C^n\over n!} \hE\Big[\sup_{0\le t\le 1} |B_t|^{n\a} + \sum_{m=0}^\infty \sup_{2^m \le t \le 2^{m+1}} {|B_t|^{n\a}\over (1+t)^{n(2+\a)\over 4}}\Big]\\
&&\le\sum_{n=0}^\infty{C^n\over n!} \hE\Big[\sup_{0\le t\le 1} |B_t|^{n\a} + \sum_{m=0}^\infty  2^{-{mn(2+\a)\over 4}}\sup_{0 \le t \le 2^{m+1}} |B_t|^{n\a}\Big]\\
&&=\sum_{n=0}^\infty{C^n\over n!} \hE\Big[\sup_{0\le t\le 1} |B_t|^{n\a} + \sum_{m=0}^\infty  2^{-{mn(2+\a)\over 4} + {(m+1)n\a\over 2}}\sup_{0 \le t \le 1} |B_t|^{n\a}\Big]\\
&&\le \sum_{n=0}^\infty{C^n\over n!}\hE\Big[\sup_{0\le t\le 1} |B_t|^{n\a}\Big] \sum_{m=0}^\infty 2^{-{mn(2-\a)\over 4}} \le \sum_{n=0}^\infty{C^n\over n!}\hE\Big[\sup_{0\le t\le 1} |B_t|^{n\a}\Big] \\
&&= \hE\Big[e^{C\sup_{0\le t\le 1} |B_t|^\a}\Big].
\eeaa
This implies \eqref{Mintegrable} immediately.

We now prove \eqref{rep1}. By Karatzas-Shreve \cite[Section 5.6. B]{KS}, conditional on $\{B_1=x\}$, $B$ is a Brownian bridge and its conditional law is equal to the law of $\bar X^{x}$.  Then, by Girsanov theorem,
\beaa
G(1, x) &=& \hP(X_1 \ge x) = \hE\Big[e^{\int_0^1 b(s, B_s) dB_s - {1\over 2} \int_0^1 |b(s, B_s)|^2ds} 1_{\{B_1 \ge x\}}\Big] \\
&=&\int_x^\infty {1\over \sqrt{2\pi}} e^{-{y^2\over 2}}  \hE\Big[e^{\int_0^1 b(s, B_s) dB_s - {1\over 2} \int_0^1 |b(s, B_s)|^2ds} \Big| B_1 = y\Big] dy\\
&=&\int_x^\infty {1\over \sqrt{2\pi}} e^{-{y^2\over 2}}  \hE\Big[e^{\int_0^1 b(s, \bar X^y_s) d\bar X^y_s - {1\over 2} \int_0^1 |b(s, X^y_s)|^2ds}\Big] dy.
\eeaa
This, together with the fact $d \bar X^{x}_s = x ds - \bar M_s ds + dB_s$, implies \eqref{rep1} immediately. 
\end{proof}

\medskip

\begin{proof}[Proof of Proposition \ref{prop-tail}.]
Again we shall only prove the case that $t=1, x_0=0$. 

We first show that, for the $I$ in \eqref{rep1},
\bea
\label{I'}
|I'(x)|\le C.
\eea
This, together with \eqref{rep}, implies immediately the first estimate in \eqref{Gbound}. 

Indeed, denote $\bar b(t,x) := \int_0^x b(t, y) dy$. Applying It\^o formula we have
\beaa
\bar b(1, x) = \bar b(1, \bar X^x_1) - \bar b(0, \bar X^x_0)=\int_0^1 \Big[\pa_t \bar b(t, \bar X^x_t) + {1\over 2} \pa_x b(t, \bar X_t^x)\Big] dt + \int_0^1 b(t, \bar X_t^x) d\bar X_t^x.
\eeaa
Then
\bea
\label{I2}
e^{I(x)} = \hE\Big[e^{\bar b(1, x) - \int_0^1 [\pa_t \bar b(t, \bar X^x_t) + {1\over 2} \pa_x b(t, \bar X_t^x)] dt}\Big].
\eea
Differentiating with respect to $x$ and noting that $\pa_x \bar X^x_t = t$, we have
\beaa
e^{I(x)} I'(x) =\hE\Big[e^{\bar b(1, x) - \int_0^1 [\pa_t \bar b(t, \bar X^x_t) + {1\over 2} \pa_x b(t, \bar X_t)] dt} \big[ b(1,x) -\int_0^1 t [\pa_t  b(t, \bar X^x_t) + {1\over 2} \pa_{xx} b(t, \bar X_t)] dt\big] \Big].
\eeaa
This implies
\beaa
e^{I(x)} |I'(x)| \le C\hE\Big[e^{\bar b(1, x) - \int_0^1 [\pa_t \bar b(t, \bar X^x_t) + {1\over 2} \pa_x b(t, \bar X_t)] dt} \Big] = Ce^{I(x)},\quad\mbox{and thus}\quad |I'(x)|\le C.
\eeaa

We next verify the second part of \eqref{Gbound} for $x>0$. The case $x<0$ can be proved similarly. Clearly it suffices to verify it for $x$ large. Note that
\beaa
{G(1,x)\over \rho(1, x)} = \int_0^\infty {\rho(1, x+y)\over \rho(1,x)} dy = \int_0^\infty e^{I(x+y)-{1\over 2}(x+y)^2 + {x^2\over 2} - I(x)}dy = \int_0^\infty e^{I(x+y)  - I(x)- xy -{1\over 2}y^2 }dy .
\eeaa
Then, for $x>C+1$, where $C$ is the bound of $I'$, 
\beaa
{G(1,x)\over \rho(1, x)} &\le& \int_0^\infty e^{Cy - xy }dy = {1\over x-C}  \le 1;\\
{G(1,x)\over \rho(1, x)}  &\ge&   \int_0^1 e^{ -Cy- xy -{1\over 2}y^2}dy \ge e^{-{1\over 2}}{1-e^{-x-C}\over x+C} \ge {c\over x},
\eeaa
completing the proof.  
\end{proof}

\medskip

\begin{proof}[Proof of Proposition \ref{prop-Grho}.]
The convergence of $G^N$ is standard, and is also implied by the convergence of $\rho^N$, so we shall only prove the latter. Assume for simplicity that $T=1$. Note that $\rho$ is locally uniformly continuous in $(0, T]\times \hR$.  Without loss of generality we shall only estimate $|\rho^N(1, x) - \rho(1,x)|$ for  $x$ in the range of $X^N_1$.  We remark that we shall assume $|x|\le R$ for some constant $R>0$, and in the proof below the generic constant $C$ may depend on $R$.

Let $\xi^N_i$, $i=1,\ldots, N$ be i.i.d. with $\hP(\xi^N_i = {1\over \sqrt{N}}) = \hP(\xi^N_i = -{1\over \sqrt{N}}) = {1\over 2}$, $B^N_{t_0} = 0$, $B^N_{t_{i+1}} := B^N_{t_i} + \xi^N_{i+1}$, and denote $b^N_i:= b(t_i, B^N_{t_i})$. Introduce the conditional expectation:
\beaa
\hE_x[\cdot] := \hE\big[\cdot|B^N_1 = x\big].
\eeaa
Then we see that
\beaa
\rho^N(1,x) &=&\hP(X^N_1 = x)\slash (2\sqrt{h}) = \hE\Big[\Pi_{i=0}^{N-1} [1+ b^N_i \xi^N_{i+1}] 1_{\{B^N_1 = x\}}\Big] \slash (2\sqrt{h})\nonumber\\
&=& \hE_x\Big[\Pi_{i=0}^{N-1} [1+b^N_i \xi^N_{i+1}] \Big] \hP(B^N_1=x)\slash (2\sqrt{h})\nonumber\\
&=&\hE_x\Big[ e^{\sum_{i=0}^{N-1} [b^N_i \xi^N_{i+1} - {1\over 2}|b^N_i|^2h ]}[1+  o(1)]  \Big] \hP(B^N_1=x)\slash (2\sqrt{h}).
\eeaa
One can easily show that $\lim_{N\to\infty}\hP(B^N_1 = x) \slash (2\sqrt{h}) = {1\over \sqrt{2\pi}} e^{-{x^2\over 2}}$, by an elementary argument using Stirling's approximation. Then it remains to establish the limit
\bea
\label{Iconv}
\hE_x\Big[ e^{\sum_{i=0}^{N-1} [b^N_i \xi^N_{i+1} - {1\over 2}|b^N_i|^2h ]}  \Big] \to e^{I(x)}.
\eea
We proceed in three steps, and for simplicity we assume $N=2n$ and $x = {2k\over \sqrt{2n}}$.

{\it Step 1.}  Fix $t\in (0, 1)$ and assume $t=t_i$ for some  even $i=2m$ (more rigorously we shall consider $t_{2m}\le t<t_{2m+2}$). For any bounded and smooth test function $f$,
\beaa
\hE_x[f(B^N_{t_{i}})] &=& \sum_j f(x_{ij}) {\hP(B^N_{t_i} = x_{ij}, B^N_1 = x)\over \hP(B^N_1=x)}=\sum_j f(x_{ij}) {\hP(B^N_{t_i} = x_{ij}, B^N_1- B^N_{t_i} = x - x_{ij})\over \hP(B^N_1=x)}\\
&=& \sum_l f(2l\sqrt{h}) {\hP(B^N_{t_i} = 2l \sqrt{h}) \hP(B^N_1- B^N_{t_i} = 2(k-l) \sqrt{h})\over \hP(B^N_1=2k\sqrt{h})}.
\eeaa
Note that ${m\over n} = t$, ${k\over n}= x\sqrt{h}$, and  denote $y := {2l\over \sqrt{2n}} = 2l\sqrt{h}$. By Stirling's formula we have
\beaa
&&\hE_x[f(B^N_{t_{i}})] = \sum_l f(2l\sqrt{h}) {{(2m)!\over (m+l)! (m-l)!} {(2n-2m)!\over (n-m +k-l)! (n-m -k +l)!}\over {(2n)!\over (n+k)! (n-k)!} }\\
&&= [1+o(1)]\sum_l f(2l\sqrt{h}) \sqrt{2m(n-m)(n^2-k^2) \over 2\pi n (m^2-l^2)((n-m)^2-(k-l)^2)} \times\\
&&\quad {m^{2m} (n-m)^{2(n-m)} (n+k)^{n+k}(n-k)^{n-k}\over (m+l)^{m+l} (m-l)^{m-l}(n-m+k-l)^{n-m+k-l}(n-m-k+l)^{n-m-k+l} n^{2n}}\\
&&= [1+o(1)]\sum_l f(2l\sqrt{h}) \sqrt{2t(1-t)(1- x^2h) \over 2\pi n (t^2- y^2h)((1-t)^2-(x-y)^2h)} {A_1\over  A_2 A_3},
\eeaa
where
\bea
\label{A123}
A_1 &:=& (1+x\sqrt{h})^{n(1+x\sqrt{h})} (1-x\sqrt{h})^{n(1-x\sqrt{h})};\nonumber\\
A_2 &:=&  (1+{y\over t}\sqrt{h})^{n(t+y\sqrt{h})}(1-{y\over t}\sqrt{h})^{n(t-y\sqrt{h})}; \\
A_3 &:=&  (1+{x-y\over 1-t}\sqrt{h})^{n(1-t+(x-y)\sqrt{h})}(1-{x-y\over 1-t}\sqrt{h})^{n(1-t-(x-y)\sqrt{h})}.\nonumber
\eea
Note that, for any $0<z<1$,
\beaa
e^{z^2} \le (1+z)^{1+z}(1-z)^{1-z} \le e^{z^2 + {2z^3\over 3}}.
\eeaa
Then, noting that $n = {1\over 2h}$,
\bea
\label{A123est}
\left.\ba{c}
 {A_1 \over A_2A_3} \le e^{n[x^2h + {2\over 3}x^3h^{3\over 2} - {y^2 h\over t} - {(x-y)^2 h\over 1-t}]} = e^{-{(tx-y)^2\over 2 t(1-t)}  +{1\over 3} x^3 \sqrt{h}};\\
 {A_1 \over A_2A_3} \ge e^{n[x^2h  - {y^2 h\over t}- {2|y|^3h^{3\over 2}\over 3t^3} - {(x-y)^2 h\over 1-t} - {2|x-y|^3h^{3\over 2}\over 3(1-t)^3}]} = e^{-{(tx-y)^2\over 2 t(1-t)}  -{1\over 3} [{|y|^3\over t^3} + {|x-y|^3 \over (1-t)^3}] \sqrt{h}}.
\ea\right.
\eea
Then, by denoting $x\approx y$ as $x=y[1+o(1)]$ for $h\to 0$, we have 
\beaa
\hE_x[f(B^N_{t_{i}}) ] \approx  \sum_l f(2l\sqrt{h})   {2\sqrt{h}\over \sqrt{2\pi  t(1-t)}} e^{-{(tx-y)^2\over 2 t(1-t)}  } \approx \int f(y) {1\over \sqrt{2\pi  t(1-t)}} e^{-{(tx-y)^2\over 2 t(1-t)}  } dy.
\eeaa
That is, for $t_i=t$, the conditional law of $B^N_t$ given $B^N_1=x$ asymptotically has density ${1\over \sqrt{2\pi  t(1-t)}} e^{-{(tx-y)^2\over 2 t(1-t)}  } dy$, which is exactly the density of the $\bar X^x_t$ defined in \eqref{rep1}.

\medskip

{\it Step 2.} Again assume for simplicity that $i=2m$ is even. Note that, for each $l$,
\beaa
&&\hP(\xi^N_{i+1} = \sqrt{h} |B^N_{t_i} = 2l\sqrt{h}, B^N_1=x) \\
&&=\hP(\xi^N_{i+1} = \sqrt{h} |B^N_{t_i} = 2l\sqrt{h}, B^N_1-B^N_{t_i}=2(k-l)\sqrt{h})  \\
&&={\hP(B^N_{t_i} = 2l\sqrt{h}, \xi^N_{i+1} = \sqrt{h}, B^N_1-B^N_{t_{i+1}}=(2k-2l-1)\sqrt{h}) \over \hP(B^N_{t_i} = 2l\sqrt{h}, B^N_1-B^N_{t_{i}}=(2k-2l)\sqrt{h})}\\
&&={\hP( \xi^N_{i+1} = \sqrt{h}) \hP(B^N_1-B^N_{t_{i+1}}=(2k-2l-1)\sqrt{h}) \over \hP(B^N_1-B^N_{t_{i}}=(2k-2l)\sqrt{h})}\\
&&= {\Big(\left.\ba{c}2n-2m-1\\ n-m+k-l-1\ea\right.\Big)\over (\left.\ba{c} 2n-2m \\ n-m+k-l\ea\right.)} = {n-m+k-l\over 2(n-m)}.
\eeaa
Note further that, given $B^N_{t_i}$,  $(B^N_{t_1},\ldots, B^N_{t_{i-1}})$ and $(\xi^N_{i+1}, B^N_1)$ are conditionally independent. Then
\beaa
\hP(\xi^N_{i+1} = \sqrt{h} |\cF^N_{t_i}, B^N_1=x) = {n-m+k-{B^N_{t_i}\over 2\sqrt{h}}\over 2(n-m)}={1\over 2} - {B^N_{t_i}-x \over 2(1-t_i)} \sqrt{h},
\eeaa
where $\cF^N_{t_i} := \si(B^N_{t_1},\ldots, B^N_{t_i})$.  This implies
\beaa
\hE_x[\xi^N_{i+1}  |\cF^N_{t_i}] =\sqrt{h}\Big[{1\over 2} - {B^N_{t_i}-x \over 2(1-t_i)} \sqrt{h}\Big]-\sqrt{h}\Big[{1\over 2} + {B^N_{t_i}-x \over 2(1-t_i)} \sqrt{h}\Big] =-{B^N_{t_i}-x \over 1-t_i} h.
\eeaa

Now denote, for $i<N-1$,
\bea
\label{tildexi}
\bar \xi^N_{i+1} := \xi^N_{i+1}  - \hE_x[\xi^N_{i+1}  |\cF^N_{t_i}] = {1-t_i\over 1-t_{i+1}} \xi^N_{t_{i+1}} + {h\over 1-t_{i+1}}[B^N_{t_i} - x]. 
\eea
Then $\cF^N_i = \si(\bar \xi_1,\ldots, \bar \xi_i)$, and
\bea
\label{Etildexi}
|\bar \xi^N_{i+1}|\le \sqrt{h},\quad  \hE_x[\bar \xi^N_{i+1}  |\cF^N_{t_i}] =0.
\eea

By induction one can easily verify 
\bea
\label{BNtildexi}
B^N_{t_i} = x t_i + (1-t_i)  \bar M^N_{t_i},\quad \mbox{where}\quad \bar M^N_{t_i} := \sum_{j=0}^{i-1} {\bar \xi^N_{j+1}\over 1-t_j}.
\eea
By \eqref{Etildexi} we see that $\bar M^N$ is a martingale under the conditional expectation $\hE_x$, and thus
\begin{equation}
\label{EM2}
\hE_x[|\bar M^N_{t_i}|^2] = \sum_{j=0}^{i-1} {\hE_x[|\bar \xi^N_{j+1}|^2]\over (1-t_j)^2}\le  \sum_{j=0}^{i-1} {h\over (1-t_j)^2} \le \int_0^{t_i} {dt\over (1-t)^2} = {t_i\over 1-t_i}.
\end{equation}

Clearly $\ol M^n_{t_i} = {B^N_{t_i}-xt_i \over 1-t_i}$. For any $C>0$, by setting $f(y) = e^{C{y-xt_i\over 1-t_i}}$ and applying the first inequality in \eqref{A123est}, we have
\beaa
\hE_x[e^{C \ol M^n_{t_i} }] \le [1+o(1)] \int e^{C{y-xt_i\over 1-t_i}} {1\over \sqrt{2\pi  t_i(1-t_i)}} e^{-{(t_ix-y)^2\over 2 t_i(1-t_i)}  + {1\over 3} x^3\sqrt{h}} dy = [1+o(1)]e^{Ct_i \over 1-t_i}.
\eeaa
Similarly, $\hE_x[e^{-C \ol M^n_{t_i} }] \le  [1+o(1)]e^{Ct_i \over 1-t_i}$. Applying the Doob's maximum inequality on the martingale $\ol M^N$ we have: for any $l\ge 2$,
\beaa
\hE_x\Big[\sum_{0\le j\le i} |\ol M^N_{t_j}|^l\Big] \le ({l\over l-1})^l \hE_x\Big[|\ol M^N_{t_i}|^l\Big] \le C\hE_x\Big[|\ol M^N_{t_i}|^l\Big].
\eeaa
This implies
\beaa
\hE_x \Big[e^{C \sup_{0\le j\le i}  |\ol M^N_{t_j}|}\Big] &=& \sum_{l=0}^\infty {C^l\over l!} \hE_x\big[\sup_{0\le j\le i}  |\ol M^N_{t_j}|^l\big] \le \sum_{l=0}^\infty {C^l\over l!} \hE_x\big[|\ol M^N_{t_i}|^l\big] \\
&=& \hE_x\big[e^{C|\ol M^N_{t_i}|}\big]\le  \hE_x\big[e^{C\ol M^N_{t_i}} + e^{-C\ol M^N_{t_i}} \big] \le Ce^{Ct_i \over 1-t_i}.
\eeaa
Now following the arguments for \eqref{Mintegrable}, one can  show that,
\beaa
\hE_x\Big[e^{Ch \sum_{i=1}^{N-1} |\bar M^N_{t_i}|} \Big] \le C.
\eeaa
Moreover, note that
\beaa
e^{\sum_{i=0}^{N-1} b_i \xi^N_{i+1}} =e^{ b_{N-1} \xi^N_N + \sum_{i=0}^{N-2} b_i \Big[\bar \xi^N_{i+1} - h \bar M^N_{t_{i+1}} + xh\Big]} \le C e^{ \sum_{i=0}^{N-2} b_i \bar \xi^N_{i+1}} e^{Ch \sum_{i=1}^{N-1} |\bar M^N_{t_i}|}.
\eeaa
By \eqref{Etildexi}  one can easily show that $\hE_x\Big[e^{C \sum_{i=0}^{N-2} b_i \bar \xi^N_{i+1}} \Big] \le C$.
Then we have
\bea
\label{MNintegrable}
\hE_x\Big[ e^{C\sum_{i=0}^{N-1} [b^N_i \xi^N_{i+1} - {1\over 2}|b^N_i|^2h ]} \Big] \le C.
\eea

{\it Step 3.}  Fix $m$ and set $s_j = {j\over m}$.  Similar to Step 1, we see that the conditional law of $(B^N_{s_1}, \ldots, B^N_{s_m})$ given $B^N_1 = x$ is asymptotically equal to the law of $(\bar X^x_{s_1}, \ldots, \bar X^x_{s_m})$.  Assume for simplicity that $N = nm$ (more rigorously we shall consider $nm \le N < (n+1)m$. Then
\begin{equation}
\label{Iconv2}
\hE_x\Big[e^{\sum_{j=0}^{m-1} [ b^N_{nj} (B^N_{s_{j+1}}-B^N_{s_j}) - {1\over 2m} |b^N_{nj}|^2 ]} \Big]  \approx \hE\Big[e^{\sum_{j=0}^{m-1} [ b(s_j, \bar X^x_{s_j}) (\bar X^x_{s_{j+1}}-\bar X^x_{s_j}) - {1\over 2m} |b(s_j, \bar X^x_{s_j})|^2 ]} \Big].
\end{equation}
Send $m\to \infty$, clearly the right side of \eqref{Iconv2} converges to $e^{I(x)}$. 

It remains to estimate the difference between the left side of \eqref{Iconv} and that of \eqref{Iconv2}. Denote
\bea
\label{dmN}
\left.\ba{c}
 \delta^N_{m,1} :=  \hE_x\Big[\Big| \sum_{j=0}^{m-1} \sum_{i=0}^{n-1} {h\over 2} \big| |b^N_{t_{nj +i}}|^2- |b^N_{t_{nj}}|^2]\big|\Big],\\ \delta^N_{m,2}:=  \hE_x\Big[\Big| \sum_{j=0}^{m-1} \sum_{i=0}^{n-1}\Big[[ b^N_{t_{nj +i}}- b^N_{t_{nj}}]  (B^N_{t_{nj+i+1}}-B^N_{t_{nj+i}})\big]\Big|\Big].
\ea\right.
\eea
For any $R>|x|$, note that $b$ is uniformly continuous on $[0, T]\times [-R, R]$ with some modulus of continuity function $\rho_R$.  Then, for $j=0,\ldots, m-1$, $i=0,\ldots, n-1$,
\begin{eqnarray} \label{dmNij}
\delta^N_{m,i,j} &:=& \hE_x\Big[\big| b^N(t_{nj +i}, B^N_{t_{nj+i}})- b^N(t_{nj}, B^N_{t_{nj}})\big|\Big]\\
&\le&C\hE_x\Big[\big[|B^N_{t_{nj+i}} -  B^N_{t_{nj}}| + \rho_R({1\over m}) + 1_{\{|B^N_{t_{nj}}|>R\}}+ 1_{\{|B^N_{t_{nj+i}}|>R\}} \big]\Big]\nonumber\\
&\le& C\rho_R({1\over m}) + {C\over  R} \hE_x\big[|B^N_{t_{nj}}| + |B^N_{t_{nj+i}}|\big] + C \hE_x\Big[|B^N_{t_{nj+i}} -  B^N_{t_{nj}}|\Big].\nonumber
\end{eqnarray}
Recalling \eqref{EM2}, we have
\beaa
&&\hE_x\big[|B^N_{t_i}|^2\big] \le C|x t_i|^2 + C(1-t_i)^2 \hE_x\big[|\bar M^N_{t_i}|^2\big] \le C(xt_i)^2 + Ct_i(1-t_i)  \le C;\\
&& \hE_x\Big[|B^N_{t_{nj+i}} -  B^N_{t_{nj}}|^2\Big]= \hE_x\Big[|[t_{nj+i}- t_{nj}][x-\bar M^N_{t_{nj+i}}] +(1-t_{nj})[\bar M^N_{t_{nj+i}} -  \bar M^N_{t_{nj}}]|^2\Big]\\
&&\le {C\over m^2} \Big[|x|^2 + \hE_x[|\bar M^N_{t_{nj+i}}|^2\big]\Big] + Ch (1-t_{nj})^2\sum_{l=0}^{i-1}{1\over (1-t_{nj+l})^2}\\
&&\le {C\over m^2}[|x|^2 + {1\over 1-t_{nj+i}}] + {C\over m} {1-t_{nj} \over 1-t_{nj+i}} \le {C\over m[1-t_{nj+i}]}.
\eeaa
Then
\bea
\label{dmNijest}
\delta^N_{m,i,j} \le C\rho_R({1\over m}) + {C\over  R} + {C\over \sqrt{m[1-t_{nj+i}}]}
\eea
Thus
\begin{equation}
\label{dmN1est}
\delta^N_{m,1} \le Ch \sum_{j=0}^{m-1}\sum_{i=0}^{n-1}\delta^N_{m,i,j} \le   C\rho_R({1\over m}) + {C\over R} + \sum_{i=0}^{N-1} { Ch\over \sqrt{m(1-t_i)}} \le C\Big[\rho_R({1\over m}) + {1\over R} + {1\over \sqrt{m}}\Big].
\end{equation}
Moreover,
\bea
\label{dmN2est}
&&\delta^N_{m,2}=  \hE_x\Big[\Big| \sum_{j=0}^{m-1} \sum_{i=0}^{n-1} [ b^N_{t_{nj +i}}- b^N_{t_{nj}}] \big[h[x-\bar M^N_{t_{nj+i+1}}] + \bar \xi^N_{nj+i+1}\big]\Big]\nonumber\\
&&\le C h\hE_x\Big[ \sum_{j=0}^{m-1} \sum_{i=0}^{n-1} |b^N_{t_{nj +i}}- b^N_{t_{nj}}| |x-\bar M^N_{t_{nj+i+1}}|\Big]+C\Big(\hE_x\Big[\Big| \sum_{j=0}^{m-1} \sum_{i=0}^{n-1} [ b^N_{t_{nj +i}}- b^N_{t_{nj}}]  \bar \xi^N_{nj+i+1}\Big|^2 \Big]\Big)^{1\over 2}\nonumber\\
&&\le C h \sum_{j=0}^{m-1} \sum_{i=0}^{n-1} \Big(\delta^N_{m,i,j}\Big)^{1\over 2} \Big(\hE_x\big[|x-M^N_{t_{nj+i+1}}||^2\big]\Big)^{1\over 2}+C\Big(h \sum_{j=0}^{m-1} \sum_{i=0}^{n-1} \delta^N_{m,i,j}\Big)^{1\over 2}\nonumber\\
&&\le  C h \sum_{j=0}^{m-1} \sum_{i=0}^{n-1} \Big(\rho_R({1\over m})+{1\over R} + {1\over \sqrt{m(1-t_{nj})}}\Big)^{1\over 2} \Big(x^2+{1\over 1-t_{nj+i+1}}\Big)^{1\over 2}\nonumber\\
&&\quad+C\Big(h \sum_{j=0}^{m-1} \sum_{i=0}^{n-1}\big[ \rho_R({1\over m})+{1\over R} + {1\over \sqrt{m(1-t_{nj})}}\big]\Big)^{1\over 2}\nonumber\\
&&\le C\Big( \rho_R({1\over m})+{1\over R} + {1\over \sqrt{m}}\Big)^{1\over 2}.
\eea

We now estimate the desired difference between  \eqref{Iconv} and  \eqref{Iconv2}. Denote
\beaa
\xi_1 := \sum_{i=0}^{N-1} [b^N_i \xi^N_{i+1} - {1\over 2}|b^N_i|^2h ],\quad \xi_2:= \sum_{j=0}^{m-1} [ b^N_{nj} (B^N_{s_{j+1}}-B^N_{s_j}) - {1\over 2m} |b^N_{nj}|^2 ].
\eeaa
Then, by \eqref{dmN}, \eqref{dmN1est} and \eqref{dmN2est},  we have
\beaa
\hE_x[  |\xi_1 - \xi_2| ] \le C[\delta^N_{m,1}+\delta^N_{m,2}] \le C\Big( \rho_R({1\over m})+{1\over R} + {1\over \sqrt{m}}\Big)^{1\over 2}.
\eeaa

Moreover, similar to \eqref{MNintegrable}, we have
\beaa
\hE_x\Big[ e^{C\xi_1} + e^{-C\xi_1} + e^{C\xi_2} + e^{-C\xi_2}\Big] \le C.
\eeaa
One can easily check that $|e^z-1| \le C\sqrt{|z|} [e^{2z} + e^{-2z}]$.  Then
\beaa
&&\Big|\hE_x\big[ e^{\xi_1}  \big] -  \hE_x\big[e^{\xi_2} \big]\Big|= \hE_x\Big[e^{\xi_2}\big|e^{\xi_1-\xi_2} -1\big|\Big] \le C\hE_x\Big[e^{\xi_2}\sqrt{|\xi_1-\xi_2|} [e^{2[\xi_1-\xi_2]} + e^{2[\xi_2-\xi_1]}\Big] \\
&&\le C\Big(\hE_x[|\xi_1-\xi_2|]\Big)^{1\over 2} \Big(\hE_x\Big[ e^{4\xi_1 - 2\xi_2} + e^{5\xi_2-4\xi_1}\Big]\Big)^{1\over 2} \le C \Big( \rho_R({1\over m})+{1\over R} + {1\over \sqrt{m}}\Big)^{1\over 4}.
\eeaa
By first sending $m\to \infty$ and then $R\to \infty$, we obtain the desired convergence.  
\end{proof}

%
%
%

\section*{Acknowledgments.}
The first and third authors are supported in part by NSF grant \#DMS-1908665. We would like to express our sincere gratitude to the anonymous reviewers for their very careful reading and many constructive suggestions which  helped us improve the paper greatly.


\bibliographystyle{plain} 
\bibliography{distortion} 


\end{document}